[1994/12/01]
\documentclass[letterpaper, 11pt]{article}
\title{Improved Integrality Gap in 
  Max-Min Allocation, \\
or, Topology at the North Pole\footnote{A preliminary version of this
  paper appeared in the proceedings of the {\em Symposium on Discrete
  Algorithms (SODA) 2023).}}}
\author{Penny Haxell\thanks{Department of Combinatorics and
   Optimization, University of Waterloo, Waterloo, ON,
   Canada. Partially supported by NSERC and by a
   Mercator Fellowship of the Research Training Group {\em Facets of Complexity}.}
 \quad Tibor Szab\'{o}\thanks{Freie Universit\"at Berlin, Berlin,
  Germany. Research partially funded by the Deutsche
   Forschungsgemeinschaft (DFG, German Research Foundation)
  under Germany's Excellence Strategy – The Berlin Mathematics 
  Research Center MATH+ (EXC-2046/1, project ID: 390685689) and the
  Research
  Training Group \emph{Facets of Complexity.}}}
\date{\today}

\usepackage{placeins}

\usepackage{amsmath, amsthm, amssymb, tikz}
\usepackage[margin=20pt,small,bf]{caption}
\usetikzlibrary{arrows, positioning, shapes}

\newtheorem{thm}{Theorem}[section]

\newtheorem{lem}[thm]{Lemma}
\newtheorem{prop}[thm]{Proposition}
\newtheorem{obs}[thm]{Observation}

\newtheorem*{claim}{Claim}
\newtheorem{fact}{Fact}

\theoremstyle{definition}

\theoremstyle{remark}

\numberwithin{equation}{section}

\newcommand{\cA}{\mathcal{A}}
\newcommand{\cB}{\mathcal{B}}
\newcommand{\cC}{\mathcal{C}}

\newcommand{\cI}{\mathcal{I}}
\newcommand{\cJ}{\mathcal{J}}

\newcommand{\abs}[1]{\left|#1\right|}

\newcommand{\explode}{\divideontimes}

\newcommand{\veps}{\varepsilon}

\begin{document}
\maketitle
\markboth{Topology at the North Pole}
{Topology at the North Pole}
\renewcommand{\sectionmark}[1]{}

\begin{abstract}
  In the max-min allocation problem a set $P$ of
  players are to be allocated disjoint subsets of a set $R$ of
  indivisible resources, 
  such that the minimum utility among all players is maximized.
  We study the restricted variant, also known as the Santa Claus
  problem, where each resource has an intrinsic positive value, and
  each player covets a subset of the resources.
Bez\'akov\'a and Dani~\cite{bezakovadani} showed that this problem
  is NP-hard to approximate within a factor less than $2$,
  consequently a great deal of work has focused on approximate solutions.
  The principal approach for obtaining
  approximation algorithms has been via the Configuration LP (CLP) of Bansal and
  Sviridenko~\cite{bansalsviridenko}. Accordingly, there has been
  much interest in bounding the integrality gap of this CLP. 
The existing algorithms and integrality gap estimations are all based
  one way or another on the combinatorial augmenting tree
argument of Haxell~\cite{haxell} for finding perfect matchings in certain
hypergraphs.

Our main innovation in this paper is to introduce the use of topological
methods, to replace the combinatorial argument of~\cite{haxell}
for the restricted max-min allocation problem. This approach yields substantial
improvements in the integrality gap of the CLP. In particular we
improve the previously best known bound of
$3.808$ to $3.534$. We also study the $(1,\varepsilon)$-restricted
version, in which resources can take only two values, and improve the
integrality gap in most cases.
Our approach applies a criterion of Aharoni and Haxell, and Meshulam, for the
existence of independent transversals in graphs, which involves
the connectedness of the independence complex. This is complemented by 
a graph process of Meshulam that decreases the connectedness of the
independence complex in a controlled fashion and hence, tailored
appropriately to the problem, can verify the criterion.  
In our applications we aim to establish the flexibility of the
approach and hence argue for it to be a potential asset in other
optimization problems involving hypergraph matchings.

\end{abstract}

\thispagestyle{empty}

\newpage

\setcounter{page}{1}

\section{Introduction} \label{sec:introduction}

In this paper
we consider the {\em restricted max-min allocation} problem.
An instance ${\cal I} = (P, R , v, \{L_p : p\in P\})$ of the problem
  consists of 
  a set $P$ of players, a set $R$ of indivisible resources,
  where each resource $r \in R$ has an intrinsic positive value
$v_r >0$, and each $p\in P$ covets a set $L_p\subseteq R$ of resources.
An {\em allocation} of the resources is a function $a: P\rightarrow
2^R$, with $a(p)\subseteq L_p$ for each $p\in P$, such that every resource is allocated to (at most) one player, that is $a (p)
\cap a(q) =\emptyset$ for every $p\neq q$.
The {\em min-value of allocation $a$} is $\min_{p\in P} v(a(p))$,
where for a set $S\subseteq R$ of resources
$v(S) = \sum_{r\in S}v_r$ represents the total value of $S$.
The objective is to maximize the min-value over all allocations of
resources. This value will be denoted by $OPT = OPT({\cal I})$.

The choice of a max-min
objective function is arguably a good one for
achieving overall individual ``Fairness'' in the
distribution of a set of indivisible resources
that are considered desirable by the players.\footnote{This is in contrast with the situation where resources are 
considered rather ``chores'', when one would usually aim to minimize the
maximum values of the subsets of resources allocated to each player.
That would be the 
setup for example in the classical makespan minimization problem, where
various jobs have to be allocated to a set of machines.} 
Since the seminal paper of Bansal and Sviridenko~\cite{bansalsviridenko}, the restricted max-min allocation
problem often goes under the name {\em Santa Claus Problem},
where the players represent children, and the resources are presents
to be distributed by Santa Claus. One imagines each
present $r$ having a ``catalogue'' value $v_r$, but some presents may
not be interesting to some children.\footnote{... since perhaps
 they already secured the
latest edition of their favorite smartphone for their birthday.} To be
fair\footnote{... and to avoid criticism from jealous parents},
Santa might wish to distribute
the presents so that the smallest total value received by any child is
as large as possible. 

The problem of how to find an optimal solution efficiently
was studied first in the special case when $L_p=R$ for
every player $p\in P$. In this case Woeginger~\cite{woeginger97} and Epstein and
Sgall~\cite{epsteinsgall} gave polynomial time approximation schemes (PTAS),
and Woeginger~\cite{woeginger00} gave an FPTAS when the number of
players is constant.
For the general case however, Bez\'akov\'a and Dani~\cite{bezakovadani}
showed that the problem is hard to approximate up to any factor $< 2$. 
On the positive side, there has been a great deal of progress towards finding
good approximations. In \cite{bezakovadani} an approximation
ratio of $|R|-|P|+1$ is achieved, as well as an additive approximation
algorithm using the standard assignment LP relaxation of the
problem. This finds a solution of value at least $T_{ALP}-\max_{r\in R} v_r$,
where $T_{ALP}$ is the optimal value of the assignment LP.
This algorithm however does not offer any approximation factor
guarantee when
$\max_{r\in R} v_r$ is large. 

To address the fact that the assignment LP can have arbitrarily large
integrality gap in general, 
Bansal and Sviridenko~\cite{bansalsviridenko} introduced the important
innovation of using a stronger LP, called the configuration LP for the
problem, which we now describe.
Given a problem instance ${\cal I}$ and $T\geq 0$, for each player $p\in P$
we define the family ${\cal
  C}_p(T) = \{ C \subseteq L_p: v(C) \geq T \}$ of {\em configurations for
$p$}. The {\em configuration LP} for ${\cal I}$ with {\em target} $T$ has a variable 
$x_{p,S} \geq 0$ for every player $p\in P$ and configuration
$S\in {\cal C}_p(T)$, and a constraint
$$\sum_{S\in {\cal C}_p(T)} x_{p,S} \geq 1$$ 
for every player $p\in P$ and a constraint
$$\sum_{p\in P} \sum_{S\in {\cal C}_p(T), S\ni r} x_{p,S} \leq 1$$
for every resource $r\in R$.

We will refer to this LP as CLP($T$) for ${\cal I}$.
Formally we minimize the objective function $0$, but the main point is
whether CLP($T$) is feasible. In the language of discrete optimization,
to say that CLP($T$) is feasible means that the union $\bigcup_{p\in
  P}\cC_p(T)$ of the $|P|$ hypergraphs $\cC_p(T)$ has a fractional
matching $x: \bigcup_{p\in
  P}\cC_p(T) \rightarrow [0,1]$ that has total value at least 1 on each $\cC_p(T)$.

For a given instance ${\cal I}$, let $T^* = T^*({\cal I})$ be the
maximum $T$ for which CLP($T$) is feasible. 
It is a striking fact from \cite{bansalsviridenko} that even though
CLP($T$) has exponentially many variables, $T^*$ can be
approximated up to any desired accuracy in polynomial time.
Note that any allocation for ${\cal I}$ of min-value $T'$ gives an (integer)
feasible solution to CLP($T'$). Hence $OPT\leq T^*$.
We will refer to $T^*/OPT$ as the {\em integrality gap}. Thus
  to prove the upper bound $1/\alpha$ on the integrality gap is to
  prove that, given any $T$ and fractional matching $x$ as described above,
  there exist $|P|$ disjoint sets $\{e^p\subseteq L_p:p\in P\}$ with
  $v(e^p)\geq\alpha T$ for each $p\in P$.

Using their configuration LP, Bansal and
Sviridenko~\cite{bansalsviridenko} obtained an
$O(\log \log |P|/ \log\log\log |P|)$-approximation algorithm
for the Santa Claus problem. They also formulated a combinatorial
conjecture and connected it
to the problem of finding an allocation with large min-value given a
feasible solution of CLP($T$).  Feige~\cite{feige} proved this
conjecture via repeated applications of the
Lov\'asz Local Lemma and hence established a constant 
integrality gap for the CLP. This was later made algorithmic by Haeupler,
Saha, and Srinivasan~\cite{haeupsahasrin} using Local Lemma
algorithmization, which provided the first
(huge, but) constant factor approximation algorithm for
the Santa Claus problem.

Asadpour, Feige, and Saberi~\cite{asadpourfeigesaberi} formulated the
problem in terms of hypergraph matching and proved an upper bound of
$4$ on the integrality gap of the CLP.
Via the machinery of \cite{bansalsviridenko} this result implies an efficient algorithm to estimate the value
of OPT up to a factor $(4+\delta)$. 
The approach of \cite{asadpourfeigesaberi} is based on a local search
technique introduced by Haxell~\cite{haxell}, where the corresponding procedure is not
known to be efficient. Polacek and Svensson~\cite{polaceksvensson}
modified the local search of \cite{asadpourfeigesaberi} and were able
to prove a quasi-polynomial running time for a
$(4+\delta)$-approximation algorithm. Finally, Annamalai, Kalaitzis, and
Svensson~\cite{annakalasven} managed to adapt the local search 
procedure to terminate in polynomial time, introducing several influential
novel ideas, which resulted in a polynomial time $12.33$-approximation algorithm. 
Subsequently Cheng and Mao~\cite{chengmao18} altered the algorithm
to establish a $(6+\delta)$-approximation guarantee, improving further
in \cite{chengmao19} to obtain a $(4+\delta)$-approximation
algorithm. Davies, Rothvoss, and Zhang~\cite{daviesrothvosszhang} also
gave an $(4+\delta)$-approximation algorithm, working in a more general setting, where
a matroid structure is imposed on the players.
The integrality gap of the configuration LP was further improved by
Cheng and Mao~\cite{chengmao18b} and Jansen and
Rohwedder~\cite{jansenrohwedder} to 3.833 and then to 3.808 by Cheng
and Mao~\cite{chengmao19} by better and better analysis of the
procedure of \cite{asadpourfeigesaberi}. 

A special case of the problem, that already captures much of its
difficulty, comes from limiting the number of distinct values taken by
resources to two. In the $(1,\varepsilon)$-restricted allocation problem
resources can take only two values $1$ or $\varepsilon$, where
$0 < \varepsilon \leq 1$. The relevance of this case is also underlined
by the fact that a key reduction step in the foundational result of~\cite{bansalsviridenko} 
required an approximation algorithm for the $(1,\varepsilon)$-restricted
allocation problem for arbitrarily small $\varepsilon >0$.

Chan, Tang, and Wu~\cite{chantangwu}, extending work of
Golovin~\cite{golovin} and Bez\'akov\'a and Dani~\cite{bezakovadani},
show that approximating OPT  up to a factor less than $2$ is already
NP-hard for the $(1,\varepsilon)$-restricted problem, for any fixed
$\varepsilon \leq 1/2$. 
Note that when $\varepsilon =1 $, so each resource has the same value,
the problem can be solved
exactly and easily via applications of a bipartite matching
algorithm. This algorithm can also be used to give a
$1/\varepsilon$-approximation, which is better than $2$-approximation for
$\varepsilon >1/2$. In
\cite{chantangwu} it was proved that the integrality gap
of the CLP for the $(1,\varepsilon)$-restricted allocation problem is at
most $3$, for every $\varepsilon$. The paper also gives a quasipolynomial-time
algorithm that finds a $(3+4\varepsilon)$-approximation.

\subsection{Our contributions}

The existing algorithms and integrality gap estimation for the Santa
Claus problem are, one way or another, based on the combinatorial augmenting tree
argument of \cite{haxell} for finding perfect matchings in certain
hypergraphs. Many of them are sophisticated variants
of the local search technique of \cite{asadpourfeigesaberi} 
and its efficient algorithmic realization in \cite{annakalasven}.

Our main innovation in this work is to introduce the use of
topological methods for the Santa Claus problem, and replace the combinatorial
argument of~\cite{haxell}. 
This approach yields substantial 
improvements in the integrality gap of the CLP. 

Our first main result improves the integrality gap from $3.808$ to $3.534$. 
\begin{thm}\label{thm:main-intro} The integrality gap of the CLP is
  at most $\frac{53}{15}$. 
\end{thm}

For our approach we make use of a criterion of Aharoni and
Haxell~\cite{aharonihaxell} and Meshulam~\cite{meshulam2} for the
existence of independent transversals in
graphs, using the (topological) connectedness of the
independence complex. In our application we apply this to an appropriately modified
line graph of the multihypergraph of all those subsets that are
valuable enough to be potentially allocated to the players. In order to show that the connectedness
of the independence complex is large enough, we run a graph theoretic
process, which is based on a theorem of Meshulam~\cite{meshulam2}.
In the process we dismantle our line graph, but control the topological
connectedness of the independence complex throughout, to make sure
that the process runs for long enough. This necessitates that we choose
our dismantling process with care and apply intricate analysis of the underlying
structures, carefully tailored to the specifics of the problem. We
employ the dual of the CLP to certify the length of the process. 

Our approach is conceptually different from that of all previous work
on the Santa Claus problem. 
The topological theorems in the background provide an incredibly
rich family of independent sets in the modified line graph, that is 
geometrically highly structured via a triangulation of a
high-dimensional simplex. In this setting, good allocations of
disjoint sets of
resources to players correspond to multicolored simplices in the
triangulation, and the existence of such an allocation is guaranteed
by Sperner's Lemma. This is in sharp contrast to the much
simpler sparse spanning tree-like structure at the heart of the
combinatorial approach, and where a solution is found via a direct
step-by-step augmentation process.

Our general strategy to show the existence of a solution of
large minimum utility seems quite flexible and we expect it to
be a useful asset for other algorithmic
problems of interest involving hypergraph matchings.

The machinery developed for the proof of
Theorem~\ref{thm:main-intro} can also be
used to improve significantly the known results on the integrality
gap of the CLP for the $(1,\varepsilon)$-restricted allocation problem.
In the next theorem we highlight some of the main
consequences of this aspect of our work.

\begin{thm} \label{thm:smallepsilon}
Let $\varepsilon < \frac{1}{2}$ and let ${\mathcal I}$ be an instance of the 
$(1,\varepsilon)$-restricted Santa Claus problem with
maximum CLP-target $T^*:=T^*({\cal I})$.
Then the integrality gap of ${\cal I}$ is at most
$f(\frac{\varepsilon}{T^*})$, where 
$f:(0,1]\rightarrow \mathbb{R}^+$ is a function satisfying
  \begin{itemize}
  \item $f(x)<3$  unless $x=\frac16$ or $x=\frac13$,
  \item $f(x)\leq 2.75$ for all $x\in (0,\frac{1}{6}) \cup
    [\frac{2}{11}, \frac13) \cup [\frac{4}{11}, 1]$, 
    and
\item $\lim_{x\rightarrow 0}f(x) < 2.479$.
\end{itemize}
\end{thm}

One important message of this theorem is the identification 
of a couple of
specific instances that seem especially hard to crack. For example,
we would be delighted to see a $(1,1/3)$ instance with an optimal CLP
target of $1$ and no allocation of min-value $2/3$. Furthermore, we
see that as
long as $\frac{\varepsilon}{T^*}$ is not too close to either of the
two problematic values, the integrality gap is substantially below $3$. 

As observed in~\cite{chantangwu} (and also explained in the proof
of Theorem~\ref{thm:smallepsilon}), the assumption $1\leq T^*<2$
captures the challenging case of the problem. Under this assumption,
the last part tells us that the integrality gap is less than 2.479
when $\varepsilon\rightarrow 0$.
This estimate compares favorably with an instance of the
problem given in~\cite{chantangwu}, that has integrality gap $2$ for
arbitrarily small $\varepsilon$.

We remark that the restriction on $\varepsilon$ in the theorem
is not crucial since, as mentioned earlier, there is 
a simple $\frac{1}{\varepsilon}$-approximation algorithm based on
bipartite matchings, which gives an approximation ratio $\leq 2$ if
$\varepsilon\geq \frac{1}{2}$. Moreover the restriction $x\leq 1$ is
also natural as $T^* \geq \varepsilon$ whenever $T^*$ is positive.

Finally, we note that our proofs in this paper can be turned into an
algorithmic procedure that constructs an allocation with the promised
min-value, but at the moment we have no control over the running
time. Thus  our results are in the same 
spirit as those of \cite{feige,asadpourfeigesaberi,
  chantangwu,jansenrohwedder,chengmao18b, chengmao19} in which the
strongest estimate on the integrality gap did not come with a
corresponding efficient algorithm to find an allocation. 
Nevertheless, together with the machinery of
\cite{bansalsviridenko}, our work can be used to efficiently
estimate the min-value of an optimal allocation.
As an application of such a theorem we can imagine a scenario where
Santa Claus might be prone to favoritism. Having
supernatural powers and plenty of summer leisure time at his
  traditional home at the North Pole, he can certainly calculate an optimal
allocation, yet may choose a suboptimal one benefitting his favorites. 
Our Theorem~\ref{thm:main-intro} combined with \cite{bansalsviridenko}
leads to a polynomial time
algorithm that parents can use to uncover any bias
Santa might have that is more blatant than $(\frac{15}{53} -\delta)$-times
the optimum.

\subsection{Related work}
The max-min allocation problem is also widely studied in the more general
case, where different players $p$ might have different utility value
$v_{pr}$ for resource $r\in R$. The Santa Claus problem corresponds to
the case when $v_{pr} \in \{ 0, v_r\}$.
This scenario was first considered by Lipton, Markakis, Mossel, and
Saberi~\cite{liptmarkmosssabe}. 
The NP-hardness result of Bez\'akov\'a and Dani~\cite{bezakovadani} about approximating
with a factor less than $2$ is still the best known for the general case.
Bansal and Sviridenko~\cite{bansalsviridenko} showed that their CLP
has an integrality gap of order $\Omega (\sqrt{|P|})$ for the general
problem. Asadpour and Saberi~\cite{asadpoursaberi} could match this
with an $O(\sqrt{|P|} \log^3 |P|)$-approximation algorithm using the CLP.
Chakrabarty, Chuzhoy, and Khanna~\cite{chakchuzkhan} 
give an $|R|^{\varepsilon}$-approximation algorithm for any constant
$\varepsilon$, that works in polynomial time, as well as a $O(\log^{10}
|R|)$-approximation algorithm that works in quasipolynomial time. 

The special case where each resource is coveted by only two
players is interesting algorithmically. In this case
Bateni, Charikar, and Guruswami~\cite{batenicharikarguruswami} 
showed that the Santa Claus problem is NP-hard to approximate to within a
factor smaller than $2$. Complementing this, Chakrabarty, Chuzhoy, and
Khanna~\cite{chakchuzkhan} give a $2$-approximation
algorithm, even if the values are unrestricted. 
The case when resources can be coveted only by three players is shown to be
equivalent to the general case~\cite{batenicharikarguruswami}.

For the classical dual scenario of min-max allocation
Lenstra, Shmoys, and Tardos~\cite{lenstrashmoystardos} gave a 2-approximation algorithm and
showed that it is NP-hard to approximate within a factor of $3/2$.
Using a configuration LP and a local search algorithm inspired by
those developed for the Santa Claus problem,
Svensson~\cite{svensson-minmax}
managed to break the factor $2$-barrier for the integrality
gap of the restricted version of the min-max allocation problem. 
Once more, this result comes with an efficient algorithm
to estimate the optimum value up to a factor arbitrarily
close to $\frac{33}{17}$,
but not with an efficient algorithm to find such an allocation.
The approximation factor was subsequently improved to $\frac{11}{6}$
by Jansen and Rohwedder~\cite{JR-11/6}, who
later~\cite{JR-quasi} also provided an algorithm that finds
such an allocation in quasipolynomial time.

\paragraph{Organization of the paper}
In Section~\ref{sec:topology} we present our topological tools and
describe our proof strategy. In Section~\ref{sec:warmup} we
demonstrate how our method works by giving a clean proof of the
known fact
that the integrality gap is at most $4$. In Section~\ref{sec:economical} we
introduce the main innovation that makes our improvement on the
integrality gap possible, and we use it in
Section~\ref{sec:mainproof} to prove Theorem~\ref{thm:main-intro}. In the 
subsequent Section~\ref{sec:proofs} we give the proof of the main statement from Section~\ref{sec:economical}.
Finally, in Section~\ref{sec:2val} we prove
Theorem~\ref{thm:smallepsilon} on the
two-values problem. Background and intuition for the topological notions we use are provided for the interested reader in the Appendix. We also provide in the Appendix a guide to the notation and terminology used throughout Sections~\ref{sec:introduction} to~\ref{sec:2val}.

\section{Topological tools and the proof strategy}\label{sec:topology}

\subsection{The setup} \label{sec:setup}
Let ${\cal I} = (P, R , v, \{L_p : p\in P\})$ be an instance of the
Santa Claus problem and let $T\in \mathbb{R}$ be a target such that
CLP($T$) is feasible.
A subset $e\subseteq L_p$ of coveted resources of some player $p\in P$
with the property that $v(e) > \alpha T$ and $v(e') \leq \alpha T$ for
every proper subset $e'\subset e$ is called an $\alpha$-{\it
  hyperedge}. We say that $p$ is the {\it owner} of $e$ or $e$ is an
$\alpha$-hyperedge of $p$. To
indicate this we might write $e^p$ if necessary.
Note that the hypergraph consisting of all $\alpha$-hyperedges is a
multihypergraph, since the same subset $e$ may be an
$\alpha$-hyperedge of several players $p$.
For example if an $\alpha$-hyperedge $e \subseteq L_p \cap L_q$ with
$p\neq q$, we will have both $e^p$ and $e^q$ in the multihypergraph.
An allocation with min-value greater than
$\alpha T$ constitutes choosing for every player $p\in P$ an 
$\alpha$-hyperedge of $p$, 
such that they are pairwise disjoint.\footnote{We note that defining $\alpha$-hyperedges to have value {\em at least} $\alpha T$ would capture more directly the integrality gap problem. However for our proof strategy the strict inequality turns out to be more natural.}

For $\alpha \in \mathbb{R}$,
the {\em $\alpha$-approximation allocation graph} 
$H({\cal I}, T, \alpha) = H(\alpha)$ is the auxiliary $|P|$-partite
graph with vertex set
$$V(H(\alpha)) = \cup_{p\in P} V_p, \mbox{ where } V_p = \{ e^p : \mbox{ $e
  \subseteq R$ is an
  $\alpha$-hyperedge of $p$} \},$$
and edge set 
$$E(H(\alpha)) = \{e^pf^q: p\neq q,  e\cap f \neq \emptyset \}.$$
An {\em independent transversal} in a vertex-partitioned graph such as
$H(\alpha)$ is an independent set (i.e. one that induces no edges)
that is a {\em transversal}, i.e. it 
consists of exactly one vertex in each partition class. Thus a
problem instance ${\cal I}$ with feasible CLP($T$) has an
allocation with min-value greater than $\alpha T$ for some $\alpha >0$
if and only if the $\alpha$-approximation allocation graph
$H({\cal I}, T, \alpha)$ has an independent transversal.  
Hence our Theorem~\ref{thm:main-intro} is implied by the following.

\begin{thm}\label{thm:main} Let $( P , R, \{ L_p: p\in P \}, v)$
  be an instance of the Santa Claus
  problem and let $T\in \mathbb{R}$ be such that the CLP($T$) is feasible.
  Then the corresponding $\alpha$-approximation allocation graph
  $H(\alpha)$ has an independent transversal with $\alpha = \frac{15}{53}$. 
\end{thm}

\subsection{Topological tools} \label{sec:topological}
In this section we introduce the main topological tools needed and
 describe how we use them in our arguments. 
 
For a given graph $G$, let ${\cal J}(G) = \{ I \subseteq V(G) :
\mbox{$I$ is independent}\}$ be its {\em independence complex}. 
Following Aharoni and Berger~\cite{aharoniberger} we define $\eta(G)$
to be the (topological) connectedness of ${\cal J}(G)$ plus
$2$.  An advantage of this shifting by $2$
is that the formulas for the following simple properties of $\eta$
simplify (see e.g.~\cite{adamaszekbarmak, aharoniberger,
  aharonibergerziv}). (In fact Part (2) is true in much greater
generality, see e.g.~\cite{aharoniberger}, but this simple statement is all we require.)
\begin{fact} \label{fact:eta} Let $G$ be a graph.
  \begin{itemize}
  \item[(1)] $\eta (G) \geq 0$ with equality if and only if $G$ is the
    empty graph (i.e. the graph with no vertices).
  \item[(2)] If graph $G$ is the disjoint union of $G_1$ and a
    non-empty graph $G_2$ then $\eta (G) \geq \eta (G_1) + 1$. Moreover, if
    $G_2$ is a single (isolated) vertex then $\eta (G) = \infty$.
      \end {itemize}
    \end{fact}
Intuitively, $\eta (G)$ represents the smallest dimension of a
``hole'' in the geometric realization of the abstract simplicial complex
${\cal J}(G)$. For the purposes of this paper it suffices to regard
$\eta$ strictly as a graph parameter satisfying Fact~\ref{fact:eta} and
the upcoming Theorems~\ref{thm:aharoniberger} and \ref{thm:meshulam}.
However, for the interested reader we provide the formal
definition, background and some intuition in the Appendix. 
    
Our proof of Theorem~\ref{thm:main} is based on two key theorems
involving the parameter 
$\eta$. The first one provides a sufficient Hall-type condition for the
existence of independent transversals. This result was
implicit already in~\cite{aharonihaxell} and~\cite{meshulam1}, and was
first observed by Aharoni.\footnote{
According to~\cite{meshulam1}, it was
noted by Aharoni (via private communication) that the method
of~\cite{aharonihaxell} implies
Theorem~\ref{thm:aharoniberger} for line graphs  
of hypergraphs. However, the special properties of line graphs are not
essential to the proof, so this version also captures the main essence
of Theorem~\ref{thm:aharoniberger}.}
It was first stated explicitly in the   
form below in~\cite{meshulam2} (see also~\cite{aharoniberger}). 
Let $I$ be an index set and $J$ be an $|I|$-partite graph with
vertex partition $V_1, \ldots , V_{|I|}$. For a subset $U\subseteq I$
we denote by $J|_U$ the induced subgraph 
$J[\cup_{i\in U}V_i]$ of $J$ defined on the vertex set $\cup_{i\in U} V_i$. 
\begin{thm}\label{thm:aharoniberger}
Let $I$ be an index set and $J$ be an $|I|$-partite graph with
vertex partition $V_1, \ldots , V_{|I|}$. If for every subset $U\subset I$ we have
$\eta\left(J|_U\right) \geq \abs{U}$, then there is an
independent transversal in $J$.
\end{thm}
The formal resemblance of Theorem~\ref{thm:aharoniberger} to Hall's Theorem for
matchings in bipartite graphs is no coincidence: the latter is a
consequence of the former.
Indeed, for a bipartite graph $B = (X\cup Y, E)$ satisfying
Hall's Condition we can define an $|X|$-partite (simple) graph $J(B)$,
where for every $x\in X$ there is a part $V_x = \{ y^x : y\in N_B(x) \}$ and
$y_1^{x_1}y_2^{x_2}$ is an edge if and only if $y_1=y_2$. Then a matching
of $B$ saturating $X$ corresponds to an independent transversal in
$J(B)$.
For a subset $U\subseteq X$, the subgraph $J(B)|_U$ is the union of
$|N(U)|$ disjoint cliques, so $\eta(J(B)|_U) \geq |N(U)| \geq |U|$ by
Properties (1) and (2) in Fact~\ref{fact:eta} and Hall's condition.   

Our second tool is a theorem of Meshulam~\cite{meshulam2}, reformulated
in a way that is particularly well-suited for our
arguments. Let $G$ be a graph, and let $e$ be an edge of $G$.
We denote by $G - e$ the graph obtained from $G$
by deleting the edge $e$ (but not its end vertices).
We denote by $G \explode e$ the graph obtained from $G$ by removing
both endpoints of $e$ and all of their neighboring vertices.
The graph $G \explode e$ is called $G$ with $e$ \emph{exploded}.

\begin{thm} \label{thm:meshulam}
Let $G$ be a graph and let $e \in E(G)$, such that $\eta(G-e) >
\eta(G)$. Then we have that $\eta(G) \geq \eta(G \explode e) + 1$.
\end{thm}

Inspired by Meshulam's Theorem we call an edge $e$ of $G$ {\em
  deletable} if $\eta(G-e) \leq \eta (G)$ and {\em explodable} if
$\eta(G\explode e) \leq \eta (G)-1$. By the theorem, if an edge is not
deletable then it is explodable.
A {\em deletion/explosion sequence, or DE-sequence,
  starting with graph $G_{start}$} is a sequence of
operations, which, starting with $G_{start}$, in each step either deletes a
deletable edge or explodes an explodable edge in the current graph.
The {\em length} $\ell (\sigma)$ of the sequence $\sigma$ is the
number of explosions in $\sigma$.  
A DE-sequence is called a {\em KO-sequence} if its outcome is a graph with an isolated vertex.
The  following are simple yet crucial properties of DE-sequences.   

\begin{obs}\label{obs:sequence}
  Let $G$ be the outcome of a DE-sequence $\sigma$ of length
  $\ell$, starting with $G_{start}$. Then the following are true.
  \begin{itemize}
 \item[(i)] $\eta (G_{start}) \geq \eta (G) + \ell$.
\item[(ii)] If $\sigma$ is a KO-sequence  then $\eta  (G_{start}) =\infty$. 
\item[(iii)] 
For any vertex $w\in V(G)$, there is a (possibly empty) sequence of deletions starting from $G$, after which $w$ is either an isolated vertex or  some edge of $G$ incident to $w$ can be exploded. 
  
  \end{itemize}
     \end{obs}
     \begin{proof}
Part (i) follows since during performing the DE-sequence $\sigma$
   the deletion of a deletable edge does not
increase the value of $\eta$ and the explosion of an explodable
edge decreases the value of $\eta$ by at least $1$. 
For (ii), by (i) and by Fact~\ref{fact:eta}(2) we have
$\eta (G_{start}) \geq \eta (G) = \infty $.
For (iii) let us consider all edges of $G$ incident to $w$, in some order.  If we can delete all of them following the order, then $w$ becomes isolated. Otherwise after some (possibly $0$) number of deletions, the next incident edge in the order is not deletable. By Meshulam's theorem this edge is explodable.  
\end{proof}
 In Appendix~\ref{sec:example} we give a small concrete example
demonstrating how to use DE-sequences to obtain a lower bound on $\eta$.

\subsection{The proof strategy}\label{sec:strategy}
Let $T$ be such that the $CLP(T)$ of instance ${\cal I}$  with the target $T$
has a feasible solution.
Our proof strategy is to take, for our chosen $\alpha$,
the graph $H(\alpha)$ defined in Section~\ref{sec:setup}
and use Theorem~\ref{thm:aharoniberger} to derive
the existence of an independent transversal in it.

Those $\alpha$-hyperedges that contain a single resource will have a special status. 
A resource $r\in R$ is called {\em fat} if $v_r >
\alpha T$, otherwise it is called {\em thin}.
The set $F = F(\alpha) := \{ r\in R: v_r > \alpha T \}$ is the set of
fat resources. Any set $S\subseteq R$ of resources with $S\cap F =
\emptyset$ is called thin. We will in particular be speaking of {\em thin
  $\alpha$-hyperedges} and {\em thin configurations}.
Note that an $\alpha$-hyperedge is thin if and only if it contains at
least two elements. The corresponding vertices
of $H(\alpha)$ are also called thin. 
For a fat resource $r\in R$, the singleton $\{ r \}$ is called a {\em fat
$\alpha$-hyperedge}, and if $r\in L_p$ then 
$r^p$ is called a {\em fat vertex} of $H(\alpha)$. 
Each fat resource $r\in F$ corresponds to a clique $C_r :=
\{ r^p : r\in L_p \}$ in  $H(\alpha)$ which forms a component, since no other
$\alpha$-hyperedge contains $r$ (due to their minimality).

As we show next, we can shift our main focus to the subgraph
$J(\alpha) := H(\alpha)- \cup_{r\in F} C_r$ of $H(\alpha)$
induced by the set of thin vertices.  
To verify the condition  of Theorem~\ref{thm:aharoniberger} we need to
consider an arbitrary subset $U\subseteq P$ of the players and 
the corresponding induced subgraph $H(\alpha)|_U$ of $H(\alpha)$.
By Fact~\ref{fact:eta}(2) the disjoint clique components corresponding to fat
vertices $r\in F_U : = F\cap (\cup_{p\in U} L_p)$ each contribute at least one to the
value of $\eta(H(\alpha)|_U)$. 
We thus need to prove that for the remaining
graph we have $\eta (J(\alpha)|_U) \geq |U|-|F_U|$.

To that end, starting with $G_{start}=J(\alpha)|_U$ we will specify a 
DE-sequence $\sigma$ and prove that either $\sigma$ is a KO-sequence
or $\ell(\sigma) \ge |U|-|F_U|$. In the former case
Observation~\ref{obs:sequence}(ii) implies
$\eta (J(\alpha)|_U) = \infty$. In the latter case, denoting by $G_{end}$ the final graph of $\sigma$,  Observation~\ref{obs:sequence}(i) and Fact~\ref{fact:eta}(1) imply
$\eta (J(\alpha)|_U) \geq \eta(G_{end}) + |U|-|F_U| \geq |U|-|F_U|$.
In both cases we have that $$\eta(H(\alpha)|_U) \geq \eta (J(\alpha)|_U)
+ |F_U| \geq |U|,$$
so the condition of Theorem~\ref{thm:aharoniberger}
is verified. Hence there exists an independent transversal in
$H(\alpha)|_U$ and we are done.
We have just proved the following.
\begin{thm}\label{thm:diet}
  Let ${\cal I} = ( P ,R, v, \{ L_p: p\in P \})$ be a problem instance
  and $T \in \mathbb{R}$ such that CLP($T$) has a feasible solution. 
 Suppose for every $U\subseteq P$ there exists a
 DE-sequence $\sigma$ starting with
  $G_{start}=J(\alpha)|_U$ such that either $\sigma$ is a KO-sequence, or
  $\ell(\sigma)\geq |U|-|F_U|$. Then $H(\alpha)$ has an independent
  transversal.
\end{thm}
  
We remark that this approach to proving the existence of an
independent transversal using $\eta$ was described in terms of a game
in~\cite{aharonibergerziv}, and used in many settings, see
e.g. \cite{aharbergkotlziv, aharonibergersprussel,
   aharholzhowaspru, haxell2, haxell3, haxellnarins}.

With Theorem~\ref{thm:diet} we have reduced our task to constructing,  
for every $U\subseteq P$, a DE-sequence $\sigma$ starting with
  $G_{start}=J(\alpha)|_U$ such that either $\sigma$ is a KO-sequence, or
  $\ell(\sigma)\geq |U|-|F_U|$.  
To prove lower bounds on the length of a DE-sequence $\sigma$ that
starts with $J(\alpha)|_U$, 
we will maintain a cover $W\subseteq R$ of all $\alpha$-hyperedges that
correspond to vertices of $J(\alpha)|_U$, that
disappeared during explosions of $\sigma$, and control
the size of $W$.
If we are able to do this, then the complement of $W$ is large,
allowing us to find an $\alpha$-hyperedge in it and hence
extend the DE-sequence further. 
Note that deletions do not remove any vertices of $J(\alpha)|_U$.

More generally, we say $W$ is a {\em cover of the DE-sequence $\sigma$}
starting with a subgraph $G_{start} \subseteq J(\alpha)|_U$ and
ending with $G_{end}$ if 
\begin{itemize}
\item[($\star$)]  every vertex $e^p$ of
$G_{start}$ with $e\cap W = \emptyset$ is present in $G_{end}$.
\end{itemize}
The natural choice to cover the $\alpha$-hyperedges corresponding
to vertices that disappeared from $G_{start} \subseteq J(\alpha)|_U$ 
during the explosions in a DE-sequence $\sigma$ is $\bigcup (e\cup f)$,
where the union is over all edges $e^pf^q$ of $G_{start}$
exploded in $\sigma$. This will be called the {\em basic cover} of
$\sigma$. Note that for the basic cover $W_{\sigma}$,
every vertex $h^s$ of $G_{start}$ with $h\cap W_{\sigma} =
\emptyset$ is unaffected by each explosion that happened during
$\sigma$ and hence is still present in the graph $G_{end}$.

In the next subsection we will demonstrate how the simple accounting
by adding up the values of the basic covers of the explosions of an
arbitrary DE-sequence starting with $J(\alpha)|_U$ and ending with a
graph with no edges
is already sufficient to derive the existence of an allocation of
min-value greater than  $\frac{1}{4}T$.
To achieve our improved bounds in Theorem~\ref{thm:main}, in Sections~\ref{sec:economical} and
\ref{sec:mainproof} we will choose our DE-sequences
and account for their accompanying covers more carefully.

\section{The demonstration of the method} \label{sec:warmup}

In this section, we apply our method to verify
Theorem~\ref{thm:main} for the ratio $\alpha=\frac{1}{4}$. We emphasize that
this can easily be proved by using instead the combinatorial method
of~\cite{haxell}; indeed, as described in the Introduction, this approach
and intricate refinements of it
have been the basis of essentially all progress on the integrality gap
of the CLP for this problem since the pivotal paper
of~\cite{asadpourfeigesaberi}. 
The aim of this section is to re-prove the basic ratio of
$\frac{1}{4}$ using our topology-based proof strategy, to establish the
context for later refinements that we employ in the rest of the paper, and that lead to
the improved ratio $\alpha=\frac{15}{53}$ in Theorem~\ref{thm:main}.

Our setup in this section is as follows.

\bigskip
\noindent{\textbf{Setup~\ref{sec:warmup}.}}
Let ${\cal I} = (P, R , v, \{L_p : p\in P\})$ be an instance of the
Santa Claus problem and let $T\in \mathbb{R}$ be a target such that
CLP($T$) is feasible. Fix $0 < \alpha < 1$ and let $U$ be
an arbitrary subset of $P$. 
\bigskip

Our approach to proving  Theorem~\ref{thm:main} with ratio $\alpha$
will be as described in Section~\ref{sec:strategy}: for the
arbitrarily chosen subset $U\subseteq P$, we will construct a
DE-sequence $\sigma$ starting with 
  $G_{start}=J(\alpha)|_U$ such that either $\sigma$ is a KO-sequence, or
  $\ell(\sigma)\geq |U|-|F_U|$.  Then by Theorem~\ref{thm:diet} we
  will have proved Theorem~\ref{thm:main} for this choice of
  $\alpha$. 

In this section, to prove  Theorem~\ref{thm:main}
for $\alpha=\frac14$, in fact it will suffice to choose an {\em
  arbitrary} DE-sequence $\sigma$ starting 
with $G_{start}=J(\alpha)|_U$ and ending with a graph $G_{end}$ with
no edges. This is possible by repeated application of
Observation~\ref{obs:sequence}(iii). 
If $G_{end}$ contains a vertex then $\sigma$ is a KO-sequence and we
are done, so we may assume that $G_{end}$ has no vertices.
We are left to show that $\ell (\sigma) \geq |U|-|F_U|$, provided that
$\alpha=\frac14$ (which we will assume only at the end of the argument).

To estimate the value of covers,  the following definition will be useful.
A subset $s\subseteq R$ is called a {\em block} if $v(s) \leq \alpha T$. 
Note then that any proper subset of an $\alpha$-hyperedge is 
a block. 

We estimate the value of the basic cover $W_{\sigma}$ by simply adding
up estimates for the basic covers of its individual explosions. 
\begin{obs} \label{obs:basic} With the assumptions of
  Setup~\ref{sec:warmup} 
 suppose $e^pf^q$ is an explodable edge in a subgraph $G$ of $J(\alpha)|_U$.
Then the value of its basic cover $e\cup f$ is at most $3\alpha T$. 
\end{obs}
\begin{proof}
The cover $e\cup f$ is a subset of the union of three blocks:
$(e\setminus \{ x \}) \cup \{ x \} \cup (f\setminus e)$, where $x\in e$
is arbitrary. Indeed, $f\setminus e$ is a block since it is a proper
subset of $f$, and both $e\setminus\{x\}$ and $\{ x\}$ are blocks since they are
proper subsets of 
$e$. For this recall that all $\alpha$-hyperedges under consideration are
thin. Consequently  $v(e\cup f) \leq 3\alpha T$. 
\end{proof}
Hence $v(W_{\sigma}) \leq 3\alpha T\ell (\sigma)$. 

To give a lower bound on this value,   we invoke the dual DCLP($T$) of the
configuration LP for instance $\cI$ and target
value $T$. 
In DCLP($T$) there is a variable $y_p\geq 0$ for each player $p\in P$, a
variable $z_r\geq 0$ for each resource $r\in R$, and for each configuration
$S\in {\cal C}_p(T)$ there is a constraint $$y_p \leq \sum_{r \in S}
z_r.$$
The objective function, which is to be maximized, is  $$\sum_{p\in P}y_p -\sum_{r\in R}z_r.$$
We will use the dual as a convenient way to verify certain inequalities by checking the feasibility of well-chosen solutions to DCLP($T$).
Since CLP($T$) is minimization problem with objective function $0$,  weak duality amounts to the following observation. 
\begin{obs}\label{obs:weak-duality}  Let $\cI$ be an instance of the Santa Claus problem and let $T\in\mathbb{R}$ be such that the CLP($T$) for $\cI$ is feasible. If $y \in \mathbb{R}^P$ and $ z \in \mathbb{R}^R$ represent a feasible solution of the DCLP($T$) for $\cI$ then $$\sum_{p\in P}y_p - \sum_{r\in R}z_r \leq 0.$$
\end{obs}
 
The following consequence of Observation~\ref{obs:weak-duality}, as well as its more refined version (Proposition \ref{prop:refined-dual}), will be applied repeatedly throughout our paper. Here it will provide a lower bound on the value of $W_{\sigma}$.

\begin{prop}\label{prop:dual} With the assumptions of
    Setup~\ref{sec:warmup}, let $c\geq 0$ and 
  $Y\subseteq R\setminus F$ be such that $v(Y\cap S) \geq 
  c$ for every thin configuration $S \in {\cal C}_p(T)$ for $p\in U$.
  Then 
$$v(Y) \geq c (|U|-|F_U|).$$
\end{prop}

\begin{proof} We define a feasible  solution to DCLP($T$) and then
  invoke Observation~\ref{obs:weak-duality}. Let 
\begin{align*} 
y_p & = \left\{ 
\begin{array}{cc} 
0 & p\not\in U \\
c & p\in U 
\end{array}
\right.
\ \ \ 
z_r = \left\{
\begin{array}{cc}
c & r\in F_U \\
v_r & r\in Y\\
0 & \mbox{otherwise}
\end{array}
\right.
\end{align*}
To check the feasibility of the solution, let 
$S\in {\cal C}_p(T)$ be an arbitrary configuration.
If $p\not\in U$, then $y_p=0$ and the corresponding constraint holds
by the non-negativity of the $z_r$. 
If $p\in U$, then $y_p =c$. If there is a fat resource $s\in S$, then
$s \in F_U$, so $\sum_{r\in S}z_r \geq
z_s = c =y_p$. Otherwise $S\cap F = \emptyset$ and hence $\sum_{r\in S}z_r \geq \sum_{r\in
   S\cap Y} v_r \geq c = y_p$. 
So the solution is feasible. Observation~\ref{obs:weak-duality} then implies $0\geq \sum_{p\in P}y_p -\sum_{r\in R}z_r = |U|c
- (|F_U|c +\sum_{r\in Y}v_r )$ and the claim follows.
\end{proof}

Now we assume $\alpha=  \frac14$. 
To obtain a lower bound on
$v(W_{\sigma})$ we apply 
Proposition~\ref{prop:dual} with $U$, $Y=W_{\sigma}$ and $c = 3\alpha T$. 
To that end we need to check $v(S\cap W_{\sigma}) \geq 3\alpha T$ 
for every
thin configuration $S\in {\cal C}_p(T)$ with $p\in U$.
Since  $v(S) \geq T  = 4\alpha T$ 
for every configuration $S$, it is enough to verify that
$v(S\setminus W_{\sigma}) \leq \alpha T$. 
As $G_{end}$ has no vertices, Property ($\star$) of $W_{\sigma}$
implies that $R\setminus (F \cup W_{\sigma})$ should contain no $\alpha$-hyperedge 
of any $p\in U$. Consequently, for any thin configuration
$S\in {\cal C}_p(T)$ with $p\in U$,
the value of $S\setminus W_{\sigma}$ should not be large enough
for an $\alpha$-hyperedge. Hence $v(S\setminus W_{\sigma}) \leq \alpha T$ 
as needed.
Proposition~\ref{prop:dual} then implies $v(W_{\sigma}) \geq 3\alpha T 
(|U|-|F_U|)$.
Combining this with $v(W_{\sigma}) \leq 3\alpha T 
\ell(\sigma)$, we obtain
$\ell (\sigma) \geq |U| -|F_U|$ and we are done by Theorem~\ref{thm:diet}.

\section{Economical DE-sequences}\label{sec:economical}
In this section we start our proof of Theorem~\ref{thm:main} by
introducing a couple of important definitions and our main tool.
The key to improving the bound from the previous section is to improve
upon Observation~\ref{obs:basic}, whose proof amounts to saying that
any explosion has a cover that is the union of three blocks. 
For example, if the intersection of $\alpha$-hyperedges $e$ and $f$ happens
to contain at least two resources, then their explosion has a cover
that is the union of only two blocks, a clear savings over
Observation~\ref{obs:basic}.  Thus the lack of such explosions
introduces restrictions on the remaining graph $G$ and helps in
searching for not one, but perhaps a  sequence of two explosions which
has a cover that is the union of fewer than six blocks, the fewer the better,
again a savings over Observation~\ref{obs:basic}. The lack of such a
sequence introduces further restrictions that we can exploit. 

More generally,  our improvement on Section~\ref{sec:warmup} relies on finding
DE-sequences whose accounting (through their covers) is done more
economically when some of the explosions are packed together.
We use two different approaches, one based on total value (in the form of ``cheap DE-sequences") and the other based on total cardinality (``$i/j$-DE-sequences"). 

We say that a DE-sequence $\sigma$ is {\it cheap} if
there exists a cover of $\sigma$ of value at most $2\alpha T 
\ell(\sigma)$.
Note that any sequence of deletions is a cheap DE-sequence of length $0$, hence 
by Meshulam's Theorem, if there is no cheap DE-sequence starting with graph $G^*$, every edge of $G^*$ is explodable.
The example above, of an explosion $e^pf^q$ with $|e\cap f| \geq 2$, is a cheap DE-sequence of length $1$, since $e\cup f$ is the union of two blocks.
In practice we often demonstrate that a DE-sequence $\sigma$ is cheap by exhibiting 
a cover that is a subset of the union of at most $2\ell (\sigma)$ blocks.

Our second type of ``economical'' DE-sequence is based on cardinality. 
For integers $1\leq j\leq i$, a DE-sequence $\sigma$ is called
  an {\em $i/j$-DE-sequence} if it has length $j$ and a cover of
  cardinality at most $i$. In our proofs $j$ will be either 2 or 3,
  and $i$ will be $2j+1$. Since in particular an $i/j$ sequence has
  $j$ explosions and a
  cover of value at most $i\alpha T$,
  a $7/3$-sequence is
  ``more economical'' than a $5/2$-sequence, and both are ``more
  economical'' than the ``full price'' sequence used in
  Observation~\ref{obs:basic}.

Our main technical theorem tells us that, during the execution of a
DE-sequence $\sigma$, if $\sigma$ cannot be extended by a cheap DE-sequence 
(or a KO-sequence), and if some thin configuration still has
total value more than $j\alpha T$ 
outside the cover $W$ of $\sigma$, then we can
extend $\sigma$ by a $(2j+1)/j$-DE-sequence.

\begin{thm}\label{thm:jm}
Assume Setup~\ref{sec:warmup}. 
 Let $G^* \subseteq  J(\alpha)|_U$ and  $W\subseteq R \setminus F$ be a subset of resources such
that ($\star$) holds with $G_{start}=J(\alpha)|_U$ and $G_{end}=G^*$.
Let $j=2$ or $3$. Suppose there is no
  KO-sequence and no cheap DE-sequence starting with $G^*$.
  If there is a thin configuration $C\in {\cal C}_p(T)$ with
  $p\in U$ and $v(C\cap W) < (1-j\alpha)T$, 
  then there exists a $(2j+1)/j$-DE-sequence starting with $G^*$.
  \end{thm}
Hence if we cannot continue $\sigma$ with any step that improves upon
Observation~\ref{obs:basic}, every thin configuration has large
intersection with the current cover $W$. This will be the key fact
that allows us to complete our proof, which is given in the following section.
The proof of Theorem~\ref{thm:jm}
is quite intricate, and is postponed to Section~\ref{sec:proofs}.

\section{Proof of Theorem~\ref{thm:main}} \label{sec:mainproof}
 
Our setup in this section is Setup~\ref{sec:warmup}.

\begin{proof}[Proof of Theorem~\ref{thm:main}]
 As we saw in Section~\ref{sec:strategy} and again in
 Section~\ref{sec:warmup}, to prove Theorem~\ref{thm:main} with the
 ratio $\alpha$ it is
 sufficient to verify that, given Setup~\ref{sec:warmup}, there
 exists a DE-sequence $\sigma$ 
starting with $G_{start} = J(\alpha)|_U$ such that either $\sigma$ is
a KO-sequence, or $\ell(\sigma ) \geq |U|-|F_U|$. 
Here we will show that this holds for $\alpha= \frac{15}{53}$, which again will only be used in the last step of the argument.

We define $\sigma$ in four phases.
Here $G$ denotes the current graph of the sequence. 
A {\em maximal cheap DE-sequence} is one that is not a proper initial subsequence of another cheap DE-sequence. 

\begin{itemize}
\item Phase 1. {\small WHILE} a KO-sequence or
  cheap DE-sequence $\tau$ exists in $G$,
  {\small DO} perform $\tau$
\item Phase 2. 
{\small WHILE} there exists a $7/3$-DE-sequence $\tau$ in $G$, {\small DO} 
                perform $\tau$ and then perform a maximal cheap
                DE-sequence.
\item Phase 3. 
{\small WHILE} there exists a $5/2$-DE-sequence $\tau$ in $G$, {\small DO} 
                perform $\tau$, and then perform a maximal cheap DE-sequence. 
\item Phase 4. {\small WHILE} $G$ has an edge {\small DO} perform a deletion or an
  explosion in $G$.
    \end{itemize}
The DE-sequence $\sigma$ is the concatenation of all the sequences $\tau$ over all the phases, in order.    When the procedure terminates, the final graph $G_{end}$ has no edge.
    If $G_{end}$ contains a vertex, then $\sigma$ is a
KO-sequence starting with $G_{start}$, as desired.
Therefore we may assume that $G_{end}$ has no vertices. Note that, in
this case, at no time during our procedure did there exist a KO-sequence starting at the current graph $G$.
    
Let $n_1$ denote the total number of explosions performed in the cheap
DE-sequences throughout Phases 1, 2, and 3, and
$W_1$ be the union of all the covers associated to these cheap
DE-sequences. By definition of cheap DE-sequence
we know that
\begin{align}\label{eq:W1}
  v(W_1) & \leq 2\alpha T 
  n_1.
           \end{align}
For $j=2,3$, let  $n_j$ denote the number of explosions performed in
$7/3$-DE-sequences during Phase 2 and $5/2$-DE-sequences during Phase
3, respectively, and let $W_j$ be the union of their corresponding covers.
For a $(2j+1)/j$-DE-sequence the number of resources in the cover is $2j+1$
and the number of explosions is $j$. Hence
\begin{align}\label{eq:Wj-card}
  |W_{2}| & \leq \frac{7}{3}n_{2} \mbox{ and }  |W_{3}| \leq \frac{5}{2}n_{3}. 
\end{align}
For these sets it will also be useful to estimate their values. For
this, recall that each thin resource is of value
at most $\alpha T$. 
Therefore
\begin{align}\label{eq:Wj-value}
v(W_{2})
             & \leq \frac{7}{3}\alpha T 
             n_{2}  \mbox{ and } v(W_{3})
              \leq \frac{5}{2}\alpha T
              n_{3}.
\end{align}
Let $n_4$ be the number of explosions performed in Phase 4 and $W_4$
the union of the basic covers corresponding to them.
Then by Observation~\ref{obs:basic} we have
\begin{align}\label{eq:W4}
  v(W_4) & \leq  3\alpha T
  n_4.
\end{align}

The next proposition is a more refined version of
  Proposition~\ref{prop:dual}, and uses the dual DCLP($T$) in a way similar
  to that of Section~\ref{sec:warmup}. We will employ it to take
snapshots at various points 
during $\sigma$ in order to derive lower bounds involving linear
combinations of the quantities $n_j$, $j=1, 2,3,4$. 
 For each $j$, we will apply it with $Y$ being a
  cover of the initial subsequence of $\sigma$ defined up to the end
  of Phase $j$, and the lower bound on $v(Y_{\leq d}\cap S)$ will be
  ensured by Theorem~\ref{thm:jm}.

\begin{prop}\label{prop:refined-dual}
Assume the conditions of Setup~\ref{sec:warmup}. Let 
$0\leq d\leq c\leq 2d$ and  
$Y\subseteq R\setminus F$ be given, and let $Y_{\leq d} := \{ y\in Y :
  v_y \leq d \}$ and $Y_{> d} := \{ y\in Y : v_y > d \}$. Suppose that
  for every $p\in U$ and every thin configuration $S \in {\cal
    C}_p(T)$ with  
$|Y_{> d}\cap S| \leq 1$, we have
\begin{align*}
v(Y_{\leq d}\cap S) & \geq \left\{ 
\begin{array}{ll}
c       & \mbox{if $Y_{>d} \cap S = \emptyset$} \\
   c-d & \mbox{if  $|Y_{>d}\cap S| = 1$.}
 \end{array}
\right.
\end{align*}
Then 
for any partition of $Y = Y_1 \cup Y_2$ we have
$$ c|U| - c|F_U| \leq d|Y_1| + v(Y_2).$$
\end{prop}
 In fact we will apply this proposition only when $c=d$ and $c=2d$.
Note that we can recover Proposition~\ref{prop:dual}
by setting $d=c$ and using that $Y=Y_{>d} \cup Y_{\leq d}$ is a partition and $d|Y_{>d}| \leq v(Y_{>d})$. 

\begin{proof}
We define a feasible  solution to DCLP($T$) and then invoke Observation~\ref{obs:weak-duality}. Let
\begin{align*} 
y_p & = \left\{ 
\begin{array}{ll} 
0 & p\not\in U \\
c& p\in U 
\end{array}
\right.
\ \ \ 
z_r = \left\{
\begin{array}{cc}
  c & r\in F_U \\
 d & r\in Y_{> d}\\
 v_r & r\in Y_{\leq d} \\
0 & \mbox{otherwise.}
\end{array}
\right.
\end{align*}
To verify the feasibility, let $S\in {\cal C}_p(T)$ be an arbitrary configuration.
If $p\not\in U$, then $y_p=0$ and the corresponding constraint holds
by the non-negativity of the $z_r$. 

Otherwise $p\in U$ and we must check $\sum_{r\in S} z_r \geq y_p = c$. 

If there is a fat resource $s\in F\cap S \subset F_U$, then 
$\sum_{r\in S}z_r \geq z_s = c =y_p$.  
Otherwise $S\cap F = \emptyset$ and we make a case distinction based
on $|Y_{>d}\cap S|$. 
If $Y_{>d}\cap S = \emptyset$, then $\sum_{r\in
  S}z_r \geq \sum_{r\in S\cap Y_{\leq d}}z_r = v(S \cap Y_{\leq d})
\geq c$.

If $Y_{>d}\cap S = \{ s \}$, then $\sum_{r\in
  S}z_r \geq \sum_{r\in S\cap Y_{>d}}z_r + \sum_{r\in S \cap Y_{\leq d}}z_r
= z_s + v(S\cap Y_{\leq d}) \geq d + c-d = c$. 
  
Finally, if  $|Y_{>d}\cap S| \geq 2$, then $\sum_{r\in
  S}z_r \geq \sum_{r\in S\cap Y_{>d}}z_r \geq 2d \geq c$.

So in all cases the solution is feasible. 
Observation~\ref{obs:weak-duality} then implies that $$0\geq \sum_{p\in P}y_p -\sum_{r\in R}z_r = c|U| -
c|F_U| -d|Y_{>d}| - v(Y_{\leq d}).$$ 
To derive our conclusion note that 
\begin{align*}
  c|U| -c|F_U| & \leq d|Y_{>d}| + v(Y_{\leq d}) \\ & = d|(Y_1)_{>d}| +
  v((Y_1)_{\leq d})  +  d|(Y_2)_{>d}| + v((Y_2)_{\leq d}) \\
  & \leq d|(Y_1)_{>d}| +
  d|(Y_1)_{\leq d}|  +  v((Y_2)_{>d}) + v((Y_2)_{\leq d}) \\
  & = d|Y_1| + v(Y_2).
\end{align*}
\end{proof}

Corresponding to the last step of our proof in
Section~\ref{sec:warmup} where we applied Proposition~\ref{prop:dual},
here we derive our first inequality  
from Proposition~\ref{prop:refined-dual}, after Phase 2 is complete.

\begin{lem}\label{lem:phase2} Assume Setup~\ref{sec:warmup}. 
Then 
  \begin{align} \label{eq:phase2}
    |U|-|F_U| & \leq \frac{2\alpha}{1-3\alpha}
    n_1+\frac73n_2.
    \end{align}
\end{lem}
\begin{proof}
  Let $G$ be the current graph after the end of Phase 2 and set
  $W=W_1'\cup W_2$, where $W_1'$ is the union of the covers
associated with cheap DE-sequences up to the end of Phase 2.
We use Proposition~\ref{prop:refined-dual} with $U$, $c=d=(1-3\alpha) T$, and
$Y=W$. For this we only need to check for every thin configuration $S \in {\cal
  C}_p(T)$ with $p\in U$ and
$W_{>(1-3\alpha) T}\cap S =\emptyset$ that $v(W_{\leq (1-3\alpha) T} \cap S) = v(W\cap S) \geq
(1-3\alpha) T$. This follows from Theorem~\ref{thm:jm} applied with $W$, $j=3$,
$G^*=G$, and $C=S$, since after Phase 2 is
complete there is no KO-sequence, cheap DE-sequence, or $7/3$-DE-sequence
starting with $G$.
By Proposition~\ref{prop:refined-dual} we conclude that
$$(1-3\alpha) T(|U|- |F_U|) \leq v(W_1') + (1-3\alpha) T|W_2|.$$
By (\ref{eq:W1}) we have $v(W_1') \leq v(W_1)\leq 2\alpha Tn_1$ and by
(\ref{eq:Wj-card}) we have $|W_2| \leq \frac{7}{3}n_2$.
This completes the proof.
\end{proof}

In our next lemma we take a snapshot after Phase 3 and derive two
inequalities. 

\begin{lem}\label{lem:phase3} Assume Setup~\ref{sec:warmup}. 
If $\alpha \geq \frac{1}{4}$  then we have
  \begin{align}
      |U|-|F_U| & \leq \frac{\alpha}{1-3\alpha}n_1+\frac76n_2
  +\frac54n_3  \label{eq:phase31} \\
 |U|-|F_U| & \leq
  \frac{2\alpha}{1-2\alpha}n_1+\frac{7\alpha}{3(1-2\alpha)}n_{2} 
           +\frac{5\alpha}{2(1-2\alpha)}n_{3}. \label{eq:phase32}
           \end{align}
\end{lem}
\begin{proof} 
  Let $G$ be the current graph at the end of Phase 3 and
  set $W=W_1\cup W_2\cup W_3$ to be the corresponding cover.
As there is no KO-sequence, cheap DE-sequence, or 
$5/2$-DE-sequence at the end of Phase 3,
Theorem~\ref{thm:jm}, applied with $W$, $j=2$, $G^*=G$, and $C=S$,
implies $v(W\cap S) \geq (1-2\alpha)T$ for any thin configuration
$S \in {\cal C}_p(T)$ with $p\in U$.

For the first inequality
we use Proposition~\ref{prop:refined-dual} with $U$, $d=(1-3\alpha)T, c=2d$, and
$Y=W$. We need to check for every thin $S \in {\cal C}_p(T)$ with $p\in U$
and $|W_{>(1-3\alpha)T}\cap S|\leq 1$ that its value is large enough.
If $W_{>(1-3\alpha)T}\cap S =\emptyset$ then 
$v(W_{\leq (1-3\alpha)T} \cap S) = v(W\cap S) \geq (1-2\alpha)T \geq 2(1-3\alpha)T=c$ by our assumption on $\alpha$.
If $W_{> (1-3\alpha)T}\cap S = \{ s\}$ then
$v(W_{\leq (1-3\alpha)T} \cap S) = v(W\cap S) - v_s \geq (1-2\alpha)T -\alpha T=c-d$, since $s$
is thin. 

By Proposition~\ref{prop:refined-dual} we conclude that
\begin{align*}
  2(1-3\alpha)T(|U| - |F_U|) &  \leq v(W_1) + (1-3\alpha)T|W_2 \cup W_3|.
\end{align*}
Using (\ref{eq:W1}) and (\ref{eq:Wj-card}), we are done.

For the second inequality
we use Proposition~\ref{prop:dual} with $U$, $c=(1-2\alpha)T$, and
$Y=W$. Recall that by the application of Theorem~\ref{thm:jm},
we already know that $v(W\cap S) \geq
(1-2\alpha)T$ for every thin configuration 
$S \in {\cal C}_p(T)$ with $p\in U$. 
Hence by Proposition~\ref{prop:dual} we conclude that
$(1-2\alpha)T(|U|-|F_U|) \leq v(W) \leq v(W_1) + v(W_2) + v(W_3)$.
Using (\ref{eq:W1}) and (\ref{eq:Wj-value}), the inequality follows.
\end{proof}

Finally, after Phase 4, we also measure the covers. 

\begin{lem} Assume Setup~\ref{sec:warmup}. 
We then have
  \begin{align} \label{eq:phase4}
    |U|-|F_U| & \leq  \frac{2\alpha}{1-\alpha}n_1+\frac{7\alpha}{3(1-\alpha)}n_{2} 
  +\frac{5\alpha}{2(1-\alpha)}n_{3}+\frac{3\alpha}{1-\alpha}n_4.
\end{align}
\end{lem}
\begin{proof}
Set $W=W_1\cup W_2\cup W_3\cup W_4$ and apply
Proposition~\ref{prop:dual} with $U$, $c=(1-\alpha)T$, and
$Y=W$.
This is possible as after Phase 4 is
complete, there are no vertices left in the final subgraph $G_{end}$ of
$J(\alpha)|_U$. Consequently the value of resources in
$S\setminus W$ is at most $\alpha T$ for any thin configuration
$S \in {\cal C}_p(T)$ with $p\in U$. That means
$v(W\cap S) \geq (1-\alpha)T$ holds.

By Proposition~\ref{prop:dual} we then find that
$(1-\alpha)T(|U|-|F_U|) \leq v(W) \leq v(W_1) + v(W_2) + v(W_3) + v(W_4)$.
Using (\ref{eq:W1}), (\ref{eq:Wj-value}) and (\ref{eq:W4}), the inequality follows.
\end{proof}

Now if $\alpha= \frac{15}{53}  > \frac{1}{4}$, then $T = \frac{53}{15} m$ and each of the
upper bounds
$(\ref{eq:phase2})$, $(\ref{eq:phase31})$,
$(\ref{eq:phase32})$, and $(\ref{eq:phase4})$ on $|U| - |F_U|$ is true. 
It is then straightforward to check that the convex
combination of these with coefficients  
$\frac{1}{35}$, $\frac{26}{245}$, $\frac{46}{2205}$, and
$\frac{38}{45}$, respectively, implies that $$|U|-|F_U| \leq 
n_1+n_2+n_3+n_4 = \ell (\sigma).$$

This completes the proof of Theorem~\ref{thm:main}.
\end{proof}

\section{Proof of the existence of economical
  DE-sequences} \label{sec:proofs}

The main goal of this section is the proof of
Theorem~\ref{thm:jm}. Before getting to it we prove two lemmas, the
second of which represents the heart of our argument.  
  
We begin with some notation and terminology in the setting of
Setup~\ref{sec:warmup}, and for a subgraph  $G^*$ of $J(\alpha)|_U$.
When our attention is focused on the multihypergraph of
$\alpha$-hyperedges, we often refer to a vertex $e^p$ of $G^*$ as an
$\alpha$-hyperedge $e$ {\em of} $p$ or {\em owned by} $p$. When the identity of
the owner is irrelevant or already established, we often omit the
reference to the owner. In particular we also sometimes say the pair
of $\alpha$-hyperedges $e$ and $f$ are explodable, without specifying
their owners.
We also say that an $\alpha$-hyperedge $g$ of $p\in U$ {\em survives} an
explosion if the vertex $g^p$ is still present in the current graph
after the explosion. 
We say that $\alpha$-hyperedges $e$ and $f$ are {\em explodable at resource
$r$} if $e\cap f = \{ r \}$ and the pair $e$ and $f$ is explodable. 
For an element $x$ and set $e$ we write $e-x$ and
$e+x$ as shorthand for $e\setminus\{x\}$ and $e\cup\{x\}$, respectively.
We use the term {\it $\alpha$-edge} for an $\alpha$-hyperedge with
exactly two elements.

The setting in which we will apply our first lemma is as
  follows: we have begun to construct a DE-sequence starting from
  $J(\alpha)|_U$ and have reached a graph $G^*$. If we cannot continue
  our sequence ``cheaply'', then certain special conditions must hold.

\begin{lem}\label{lem:simple}
  Assume Setup~\ref{sec:warmup} and let 
$G^* = (V, E)$ be a subgraph of $J(\alpha)|_U$. Suppose there is no
KO-sequence or cheap DE-sequence starting with $G^*$.
  Let $f^q\in V$ be an $\alpha$-hyperedge of
  player $q$. Then the following hold.
  \begin{itemize}
  \item[(i)]
    If for player $p\in U$ there exists an
    $\alpha$-hyperedge $g^p \in V$ such that $g^pf^q\in E$
    then for every $\alpha$-hyperedge $e^p \in V$ of $p$ we have $|f\cap e| \leq 1$.\\
In particular, for any $e^p \in V$ with
    $|e\cap f| \geq 2$, we have $e^pf^q\not\in E$. \\
    Even more in particular, if $e^p, e^q \in V$, then $e^pe^q \not\in
    E$.
    \item[(ii)]  For every resource $r \in f$
      there exists an $\alpha$-hyperedge $g$ in $V$ 
      that is explodable with $f$ at $r$.
\end{itemize}
 \end{lem}
  \begin{proof}
    Recall that since there is no cheap DE-sequence starting with
    $G^*$, every edge of $G^*$ is explodable.

 To prove (i), first we show that if $e^p\in V$ with $|e\cap f| \geq 2$ then
 $e^pf^q\not\in E$. Indeed, as mentioned in
 Section~\ref{sec:economical}, otherwise  $e^pf^q$ is explodable and  
    $(e-x) \sqcup ((f\setminus e)+x))$, where $x\in e\cap
    f$, is a partition of the  basic cover $e\cup f$
    into two blocks.  This demonstrates the explosion of the edge
    $e^pf^q$ is cheap, a contradiction. Here
  we use that $e-x$ is a proper subset of the $\alpha$-hyperedge $e$
  since $x \in e$, and $(f\setminus e) +x$ is a proper subset of the
$\alpha$-hyperedge $f$, since $|e\cap f| \geq 2$.

Let now $g^p\in V$ be an $\alpha$-hyperedge of $p$ that is explodable
 with $f^q$. Suppose on the contrary that there exists an
 $\alpha$-hyperedge $e^p$ of $p$ such that $|e\cap f| \geq 2$.
By the above $e^pf^q\not\in E$ and $e\neq g$.
 We define a DE-sequence of length two (starting with $G^*$) and show that it is cheap,
giving a contradiction. 
First explode $f^q$ with $g^p$. Note that $e^p$ survives, since $e^pf^q$ is
not in $E$, and $g^p$ and $e^p$ are owned by the same player.
Then, since there is no KO-sequence isolating $e^p$, after
possibly some deletions, we can explode $e^p$ with some neighbor
$h$. (See Observation~\ref{obs:sequence}(iii).)
Then the basic cover $f\cup g\cup e\cup h$ of this DE-sequence of
length two has a partition $(e-x)\sqcup ((f\setminus e) +x)\sqcup (g\setminus
f)\sqcup (h\setminus e)$, where $x\in e\cap f$, into four blocks. 
Hence this DE-sequence is cheap, as claimed. This verifies the main
statement of $(i)$, which in turn implies the second statement.

For the last statement of (i) note that in our setting $V$ contains
only thin hyperedges, so $|e| \geq 2$.

For $(ii)$ note that since there is no KO-sequence starting with $G^*$, $f$ has some
neighbor in $G^*$. Once again, every edge of $G^*$ is
explodable. For a contradiction assume that
every $\alpha$-hyperedge $g$ with $f\cap g = \{ r\}$ is not
explodable with $f$. Explode $f$ with an arbitrary neighbor $h$. By $(i)$ we know that
$f\cap h=\{ s \}$ for some $s$. By our assumption $s\neq r$.
The key observation here is that the set $W^*=(f\cup h) -r$, i.e. something less than the
basic cover, is also a cover of the explosion of the edge $fh$.
For this, let $g$ be an $\alpha$-hyperedge of $G^*$ that did not
survive the explosion of the edge $fh$. If $g$ is a neighbor of $h$ in
$G^*$ then $W^*\cap g \supseteq h\cap g \neq \emptyset$, since $r\not\in
h$. Otherwise $g$ is a neighbor of $f$ and is covered by $W^*$
unless $g\cap f = \{ r\}$. However our starting assumption was that
$f$ had no such neighbor in $G^*$.
Then the cover $W^* = (f-r) \cup (h\setminus f)$ of the single explosion
of the edge $fh$ is the union of two blocks, which makes it cheap, a contradiction.
\end{proof}

Our next lemma, which will be key to the proof of
  Theorem~\ref{thm:jm}, applies in the same setting as that of
Lemma~\ref{lem:simple}. It provides much stronger consequences of the
fact that a DE-sequence starting from $J(\alpha)|_U$ with cover $W$
cannot be extended cheaply. 

\begin{lem}\label{lem:nongraph}
  Assume Setup~\ref{sec:warmup}. 
Let $G^*= (V,E)$ be a subgraph of $J(\alpha)|_U$, and let  $W\subseteq
R \setminus F$ be a subset of resources such 
that ($\star$) holds with $G_{start}=J(\alpha)|_U$ and $G_{end}=G^*$.
  Suppose there is no KO-sequence and no cheap DE-sequence starting
  with $G^* = (V,E)$. 

  Let $C\in {\cal C}_p(T)$ be a configuration of player $p\in
  U$, such that $C\setminus W$ contains an $\alpha$-hyperedge $e$
  with $|e|\geq 3$.
Then $v(C\setminus W) \leq \frac{3}{2}\alpha T$. 
\end{lem}

\begin{proof}[Proof of Lemma~\ref{lem:nongraph}]
To begin we observe that from our setup it follows that if a subset $S
\subseteq C\setminus W$ has value $v(S) > \alpha T$ then $S$ contains some
$\alpha$-hyperedge $h$ of $p$ and by property ($\star$) of $W$ we have
$h^p \in V$. Throughout this proof,
when talking about
$\alpha$-hyperedges contained in $C\setminus W$, we mean those owned
by $p$, unless otherwise specified.

To prove Theorem~\ref{lem:nongraph} we will verify that the
  assumptions imply the following
stronger statement.

\begin{claim}
Let  $s$ denote the second-most valuable element
  of $e$ (where ties are broken arbitrarily). Then $C\setminus (W+s)$
  contains no   $\alpha$-hyperedge. 
\end{claim}

To derive the conclusion of Theorem~\ref{lem:nongraph} from the Claim:
note first that 
the second-most valuable element $s$ of $e$ satisfies
$v_s\leq \frac{\alpha T}{2}$, otherwise the value of the two most valuable elements
of $e$ would exceed $\alpha T$, contradicting that $|e| \geq 3$. 
Then $v(C\setminus W) = v(C\setminus (W+s)) + v_s \leq \alpha T + \frac{\alpha T}{2}$.

The rest of the proof consists of proving the Claim.
Suppose on the contrary that $C\setminus (W+s)$ contains an
$\alpha$-hyperedge. We claim that some such $\alpha$-hyperedge $g$
intersects $e-s$ nontrivially. Let $f \subseteq C\setminus (W+s)$ be an
$\alpha$-hyperedge and suppose it is disjoint from  $e-s$. We know $f$ contains at
least two elements, so removing the least valuable element $r$ from
$f$ results in $v(f-r)\geq \frac{v(f)}{2}> \frac{\alpha T}{2}$. Similarly, since $s$ is not
the unique most valuable element of $e$ we 
find that $v(e-s)\geq \frac{v(e)}{2}>\frac{\alpha T}{2}$. Therefore $(f-r)\cup (e-s)
\subseteq C\setminus (W +s)$ is
valuable enough to contain an $\alpha$-hyperedge $g$ of $p$, and
since $f-r$ is not valuable enough to contain it, we have $g\cap (e-s) \neq
\emptyset$, as claimed. 

Recall that since $g \subseteq C\setminus W$, we know $g^p \in V$. 
Let $a$ be the most valuable element in $g\cap (e-s)$.

\medskip

\noindent{\bf Case 1.} $|g\cap (e-s)| \geq 2$. Let $b\neq a$ such that $b\in g \cap (e-s)$. 
Then $v_a\geq v_b$, so $v_s\geq v_b$ and therefore
$v((g\cup e) -b) \geq v(g - b +s) \geq v(g) > \alpha T$. Hence
$(g \cup e )- b$ contains an $\alpha$-hyperedge $h$ of $p$ and
$h^p\in V$.

We achieve a contradiction by defining a DE-sequence of length two
starting with $G^*$, which turns out to be cheap.
First explode $g^p$ at $b$ with some $\alpha$-hyperedge $d_1$ in $G^*$, which
exists by Lemma~\ref{lem:simple}(ii). By part (i) $|d_1 \cap e| \leq
1$ and hence $d_1 \cap e = \{ b \}$.
Consequently $h^p$ survives the explosion of the edge $gd_1$, since
$h\cap d_1 \subseteq ((g\cup e) -b) \cap d_1= \emptyset$ and $h$ has
the same owner as $g$.
Since there is no KO-sequence isolating $h^p$ starting with $G^*$,
after possible deletions now
we can explode $h^p$ with some $\alpha$-hyperedge $d_2$.
We claim that the basic cover of
this  DE-sequence of length two (starting with $G^*$) is a subset of the union of four blocks: 
$$d_1 \cup g \cup h \cup d_2 \subseteq (d_1\setminus g) \cup (e- b) \cup (g-a)
\cup (d_2\setminus h).$$
Here we used that $(h\setminus g) + a$ is contained in $e-b$. 
This contradiction completes Case 1.

\medskip

\noindent{\bf Case 2.} $g \cap (e- s) = \{ a \}$.

\medskip

\noindent{\bf Case 2.a.} $v((g\cup e)-a) >\alpha T$.

Let us choose an $\alpha$-hyperedge
$f\subseteq(g\cup e)-a$ of $p$, as follows. If there is a resource $b\in e-a$ such that
$v(g-a+b)>\alpha T$, choose $f\subseteq g-a+b \subseteq (g\cup e)-a$, and
otherwise choose one arbitrarily.

By Lemma~\ref{lem:simple}(ii) $g^p$ has
an explodable neighbour $d$ at $a$ in $G^*$. 
Since $g$ and $e$ are both $\alpha$-hyperedges of $p$,
by Lemma~\ref{lem:simple}(i) we have that 
$|d\cap e| \leq 1$ and consequently $d\cap e = \{ a \}$. 

We achieve a contradiction by defining a DE-sequence of length two
starting with $G^*$,
which turns out to be cheap. First we explode $g^p$ with $d$.
The $\alpha$-hyperedge
$f^p$ survives this explosion since $f\cap d \subseteq ((g\cup e)
-a) \cap d =\emptyset$ and $f$ and $g$ are both
$\alpha$-hyperedges of $p$.
Secondly, (after some possible deletions) we explode $f^p$ with some
neighbour $h$. This is possible since there is no KO-sequence
isolating $f^p$.

We claim that the basic cover $d\cup g \cup h \cup f$
of this DE-sequence of length two is the subset of 
the union of four blocks, providing the contradiction we seek.
There is a slight difference in the accounting depending how $f$ was chosen.

If $b$ is such that $v(g-a+b) > \alpha T$ and $f \subseteq g-a+b$,
we take the blocks 
$(d\setminus g)\cup (h\setminus f))\cup(g-a) \cup\{a,b\}$. 
Note that $\{ a, b\}$ is a block since it is a proper subset of the
$\alpha$-hyperedge $e$ with at least three elements.

Otherwise $f \subseteq (g\cup e)-a$ and $g-a+b$ is a block for every
$b\in e-a$. Then we take the blocks
$(d\setminus g)\cup (h\setminus f))\cup(g-a+b) \cup(e-
b)$, where $b\in e-a$ is arbitrary. Here note that
the union of the third and the fourth term is $g \cup e$, which in
turn contains $f$. This completes Case 2.a.

\medskip

\noindent{\bf Case 2.b.} $v((g\cup e)-a) \leq \alpha T$.

  Let $x$ denote the most valuable element in $g-a$ (here $\{x\}=g-a$ is
  possible). By Lemma~\ref{lem:simple}(ii) there is an
  $\alpha$-hyperedge $h$, that is explodable with $g^p$ at $x$ in $G^*$.

  \noindent{\bf Case 2.b.i.} $h\cap (e-a)= \emptyset$. 

   In this case we define a cheap DE-sequence starting with $G^*$,
giving a contradiction. We start by exploding $g^p$ with
  $h$. The $\alpha$-hyperedge $e^p$ survives this explosion,
  since $h\cap e = \emptyset$ (note that $a\not\in h$ since $a \in
  g -x = g \setminus h$)
  and $g$ and $e$ are both
  $\alpha$-hyperedges of $p$.
  Again, since there is no KO-sequence isolating $e^p$,
  after possibly some deletions, we explode $e^p$ with some
  neighbour $d$. The basic cover $g\cup h \cup e \cup d$
  of this DE-sequence is the subset of the union of four blocks:
  $(h\setminus g)\cup \{ a \} \cup ((g \cup e)-a) \cup (d\setminus
  e)$.
  This contradiction shows that this case cannot hold.

  \noindent{\bf Case 2.b.ii.} $h\cap (e-a) \neq \emptyset$.

Since $e$ is an $\alpha$-hyperedge of the owner of $g^p$, by
  Lemma~\ref{lem:simple}(i) we have $|h\cap
  e|\leq 1$, and hence $h\cap e = \{ b \}$ for some $b \neq a$. 
  To complete the proof, we will argue that the explosion of the edge
  $hg$ is cheap by establishing that the basic cover 
$h\cup g$ can be partitioned into two blocks:
$(h- b)\cup (g- x+b)$. 
The rest of the proof is concerned with
demonstrating that $g-x+b$ is indeed a block (the first term is clearly
a block).

If $g-a=\{x\}$ is a singleton, then $g-x+b = \{a, b\}$ which is a
proper subset of the $\alpha$-hyperedge $e$ of size at least three and
hence is a block. 

Otherwise fix a resource $y\in g-a-x$,
and suppose on the contrary that $v(g- x+b)>\alpha T$.
We will find an $\alpha$-hyperedge $d$ of $p$ with $h \cap d \supseteq
\{ b , x\}$, which 
would contradict Lemma~\ref{lem:simple}(i) since $h$ is explodable with an
$\alpha$-hyperedge of $p$, namely $g$.

Since $y\in g$, we have $v_x\geq v_y$, and hence for $X:=g- y+b$ 
we have  $v(X) \geq v(g-x+b)  >\alpha T$. Note however, that both $X-a$ and $X-b$ are
blocks, since $v(X-a) \leq v(g\cup e) -a) \leq \alpha T$, and $X-b=g-y$ is a
proper subset of the $\alpha$-hyperedge $g$.
Hence $a, b \in f$ for any $\alpha$-hyperedge $f\subseteq X$ of $p$.

If $x \in f$ as well, then we are done.
Otherwise let us fix an $\alpha$-hyperedge $f=\{a,b,u_1,\ldots,u_t\}
\subseteq X$ and modify it slightly
to obtain the appropriate $d$.

Note that $t\geq 1$ since $\{ a, b\}$ is a proper subset of $e$.
Then $f'= f- u_1 +x = \{a,b,x,u_2,\ldots,u_t\}$
contains an $\alpha$-hyperedge $d$ of $p$ because $v(x)\geq v(u_1)$
since $u_1 \in f-a-b \subseteq g-y-a$.
Note that since $d\subseteq
f-u_1+x\subseteq X+x$ and $x\in g-a \subseteq X$ we find
$d\subseteq X$, and so $\{a,b\}\subset d$ (since $X-a$ and $X-b$ are
blocks). 
Furthermore $d$ must also contain $x$, since otherwise 
$d \subseteq f'-x = \{a,b,u_2,\ldots,u_t\}$ is a block (as a proper
subset of $f$).
This completes the proof of Case 2.b, and that of the theorem.
\end{proof}

We are now ready to prove Theorem~\ref{thm:jm}.
 
  \begin{proof}[Proof of Theorem~\ref{thm:jm}]
    
Since $v(C\setminus W) = v(C) - v(C\cap W) > T - (1-j\alpha)T) = j\alpha T$ and $C$ is thin,
there are strictly more than $j$ resources in $C\setminus W$; let  
$s,t_1,t_2,\ldots,t_j$ be the $j+1$ most valuable ones, in increasing
order of value.
Then  $v(C\setminus (W\cup\{t_2,\ldots,t_j)\})>\alpha T$ and consequently
$C\setminus(W\cup \{t_2,\ldots,t_j\})$ must contain an
$\alpha$-hyperedge.

Since $v(C\setminus W) > j\alpha T >\frac{3}{2}\alpha T$, by
Lemma~\ref{lem:nongraph} every $\alpha$-hyperedge in $C\setminus W$
has cardinality two. So in particular the two most valuable elements
$s$ and $t_1$ of $C\setminus(W\cup \{t_2,\ldots,t_j\})$ form an
$\alpha$-edge of $p$ and $(st_1)^p \in V$.
Then, since 
$v_{t_1}\leq v_{t_i}$ for $i\in\{2,\ldots,j\}$, we also know that each $st_i$ is an
$\alpha$-edge of $p$ and $(st_i)^p \in V$.

Next we derive a couple of crucial observations about the explodable neighbors of
$st_i$ at $t_i$.
\begin{claim}
  For every $i = 1, \ldots , j$, if an
$\alpha$-hyperedge $f^q$ (of some player $q\neq p$), with $t_i\in
f$, is explodable with
$(st_i)^p$ in $G^*$ then
\begin{itemize}
\item $f$ is an $\alpha$-edge and
\item $f\cap (C\setminus W) = \{ t_i \}$.
\end{itemize}
\end{claim}
\begin{proof}
If $f$ had at least 3 resources then the explosion would be cheap, contradicting our assumption.
To see this choose $a\in f - t_i$, and observe that the basic cover
$\{ s, t_i \}\cup f$ can be partitioned into two blocks  $\{a,t_i\}$ and
$f\setminus\{a,t_i\}+s$.
Indeed, $\{a,t_i\}$ is a proper subset
of $f$ and hence is a block, and using $v_s\leq v_{t_i}$ we see that
$v(f\setminus\{a,t_i\}+s)\leq v(f\setminus\{a,t_i\}+t_i)=v(f-a)\leq \alpha T$. 
Thus every such $f$ is an $\alpha$-edge.

If $f\subseteq C\setminus W$ then by our setup ($\star$) implies that $f^p\in V$. This
contradicts Lemma~\ref{lem:simple}(i) since $f^q$ is explodable with an
$\alpha$-hyperedge owned by $p$ (that is $st_i$) and $|f\cap f| \geq
2$. So $f$ is not contained in $C\setminus W$, and since $f$ is an $\alpha$-edge the Claim follows.
\end{proof}

We now define a DE-sequence $\tau$ of length $j$ starting with $G^*$
that has basic cover $W_{\tau}$ of size $(2j+1)$. 

We construct $\tau$ by finding explosions one by one.
Let $\tau_0$ be the empty DE-sequence.
Suppose for some $i, 1\leq i \leq j$ we have already found DE-sequence
$\tau_{i-1}$ which performs explosions of the edges
$(st_1)f_1, \ldots , (st_{i-1})f_{i-1}$, in this order, such that for every
$k=1, \ldots i -1$ we have 
\begin{itemize}
\item[(a)] $f_k$ is an $\alpha$-edge and
\item[(b)] $f_k\cap (C\setminus W) = \{ t_k\}$.
\end{itemize}
Note that by $(a)$ and $(b)$ for the basic cover $W_{\tau_{i-1}} = \{ s \} \cup
\cup_{k=1}^{i-1} f_k$ we have $|W_{\tau_{i-1}}|  = 2i-1$.
 
Our first step in constructing $\tau_i$ from $\tau_{i-1}$ is to perform, iteratively,
all possible deletions, so any remaining edge in the current graph
$G$ is explodable.  
Note that $(st_i)^p$ survived all the explosions of $\tau_{i-1}$ since
it is disjoint from $f_1, \ldots , f_{i-1}$ by (b) and has the same
owner as each other $(st_j)^p$.
To complete the definition of $\tau_i$, our aim is to find an
explodable neighbor $f_i$ of $(st_i)^p$ at $t_i$ in the current
graph $G$, and explode it.

Note we cannot apply Lemma~\ref{lem:simple}(ii) here directly to show
the existence of such an $f_i$, since after executing
$\tau_{i-1}$, we do not know any more that there is no cheap
DE-sequence starting with the current graph $G$.   
So we suppose that
$(st_{i})^p$ has no explodable neighbour at $t_i$ in the current graph
$G$ and obtain a contradiction. 
Since there is no KO-sequence starting with $G^*$ that isolates
$(st_i)^p$, by Observation~\ref{obs:sequence}(iii) there must still exist an $\alpha$-hyperedge
$e$ explodable with it, and this explosion by our assumption is at $s\in e$. 
If we now perform this explosion then we claim that the resulting 
DE-sequence $\tau'$ starting with $G^*$ would be cheap, providing a
contradiction. 
Indeed, the basic cover of $\tau_{i-1}$ together with $e$ forms a cover
of $\tau'$.  (By our assumption that $(st_{i})^p$ has no explodable
neighbour at $t_i$ in $G$, we do not need to include $t_i$ in the cover.)
The value $v(W_{\tau_{i-1}})  + v(e-s) \leq  (2i-1)\alpha T +\alpha T = 2i\alpha T$ shows that
$\tau'$ is cheap. So for some player $q \neq p$
there exists an $\alpha$-hyperedge $f_i^q \in V(G)$ which is
  explodable with $(st_i)^p$ at $t_i$ in $G$.

Now we show that $f_i$ satisfies properties $(a)$ and $(b)$. To
see this observe that $f_i^q$ was also explodable with $(st_i)^p$ in
$G^*$. 
Otherwise, by Meshulam's Theorem, the edge $f_i^q(st_i)^p$ was
deletable in $G^*$. But a deletable
edge in $G^*$ is a cheap DE-sequence of length zero, contradicting our
assumption on $G^*$. 
Hence our Claim applies to $f_i$ and so (a) and (b) hold for $k=i$.

For the basic cover of the ultimate DE-sequence $\tau_j$ we have
  $|W_{\tau_{j}}|  = 2j+1$, showing that $\tau_j$ is a $(2j+1)/j$-DE-sequence.
    \end{proof}

\section{Two values}\label{sec:2val}

Here we consider the $(1,\varepsilon)$-restricted version
of the Santa Claus problem, in which the value of each resource in $R$ is either 1 or $\veps$, where $0<\varepsilon<1$. Our overall approach in this section conceptually
parallels our work in the earlier sections.

We begin with a very high-level overview of our argument. Let
  an instance $\cI=(P,R,v,\{L_p:p\in P)$ of the
  $(1,\varepsilon)$-restricted problem be 
  given, and let $T$ be such that the CLP($T$) for $\cI$ is
  feasible. Let $c=\lceil\frac{T}{\veps} \rceil$. We will define an
  integer $r_c$ and a real number $\alpha$, as functions of $T$ and $\veps$.
  It will be
  straightforward to check (see the proof of Theorem~\ref{thm:2val})
  that we may assume with these definitions that the resources of
  value 1 are fat and those of value $\veps$ are thin. Hence for each
  $p\in P$, the set of thin configurations for $p$ is precisely the
  set of subsets of $L_p$ of cardinality at least $c$.

  Our proof will consist of showing that there exists an allocation
  of disjoint sets of resources to the players in $P$, where each
  set is either a single resource of value 1 or an $r_c$-subset of
  resources of value $\veps$. Hence the integrality gap for $\cI$ is
  at most $\frac{T}{r_c\veps}\leq\frac{c}{r_c}$ (since it will also be easy to show that we may assume   $1\geq r_c\veps$ ).

Our main lemma for these purposes will apply in the following setting.

\medskip

\noindent{\textbf{Setup~\ref{sec:2val}.}}
Let ${\cal I} = (P, R , v, \{L_p : p\in P\})$ be an instance of the
$(1,\veps)$-restricted Santa Claus problem and let $T\in \mathbb{R}$
be a target such that 
CLP($T$) is feasible. Fix an integer $r$ with $2\leq r$
and suppose $\alpha\in \mathbb{R}$ satisfies $$\min\left\{r\veps, 1 \right\}  > \alpha T \geq (r-1)\veps.$$ 

Let $U$ be an arbitrary subset of $P$, let $G^* \subseteq J(\alpha)|_U$ and  $W\subseteq R \setminus F$ be a subset of resources such
that ($\star$) holds with $G_{start}=J(\alpha)|_U$ and
$G_{end}=G^*$. Suppose that there is no KO-sequence starting with $G^*$ and every edge of $G^*$ is explodable.

\medskip

We remark that the conditions on $\alpha$ are simply to ensure that resources of value $1$ are fat, 
  resources of value $\veps$ are thin, and thin $\alpha$-hyperedges
  have cardinality $r$ and value $r\veps$. 

The idea of our proof is to use the same
basic framework as that of Theorem~\ref{thm:main}. As before, we will
define a DE-sequence starting with $J(\alpha)|_U$ by concatenating
many shorter DE-sequences constructed in
phases. Each phase lasts as long as there remains a thin
configuration $C$, with value $v(C\setminus W)$ outside the current
cover $W$, that exceeds a certain threshold associated with that phase. 
However, in our current $(1,\varepsilon)$-restricted setting, the
value $v(C\setminus W)$ is determined entirely by $|C\setminus
W|$. Instead of four phases, we will have one phase for each $X, c\geq
X \geq r_c,$ which lasts as long as there still exists a thin
configuration $C$ with $|C\setminus W| \geq X$. Our notion of
''economical'' DE-sequence will be one of low {\em average cost},
where the average cost of $\sigma$ with (minimum) cover $W_{\sigma}$ is
$av(\sigma)=|W_{\sigma}|/\ell(\sigma)$.  
Each iteration in Phase $X$ finds a DE-sequence of average cost
bounded above by $a_{r_c}(X)$, where
$a_r(X)$ is the function defined for all integers $X\geq r\geq 2$ as 
\begin{align*} 
  a_r(X) & := \left\{
            \begin{array}{ll}
              3r-X-1 & r \leq X \leq \frac{3r-1}{2} \\
  2r- \frac{X+1}{3} & \frac{3r}{2} \leq X \leq 2r \\
  \frac{4r-1}{3} & 2r+1 \leq X.
\end{array}
                   \right.
                   \end{align*}
The main tool that enables this is the following lemma (which
corresponds to Theorem~\ref{thm:jm} in the proof of
Theorem~\ref{thm:main}).
We say that a DE-sequence $\sigma$ starting with $G^*$ is {\it
  based in} a 
configuration $C$ if one $\alpha$-hyperedge out of every pair of $\alpha$-hyperedges exploded
during $\sigma$ is in $C\setminus W$ (and is owned by the owner of $C$).

\begin{lem}\label{lem:basedinC}
Assume Setup~\ref{sec:2val}. 
Let $X\geq r$ be an integer. For every thin configuration $C\in {\cal
  C}_p(T)$ with 
  $p\in U$ and $|C\setminus W| \geq X$
there exists a DE-sequence $\sigma$ starting with $G^*$ based in $C$
with  $av(\sigma)\leq a_r(X)$.
\end{lem}
 The technical definition of the function $a_r(X)$ is an artefact of our proof this lemma, presented in Subsection~\ref{sec:shortDE}. A plot of $a_r(X)$ for some representative values of $r$ appears in the Appendix.

  In Section~\ref{ineq} we describe our procedure in which we execute the phases to construct a long DE-sequence. Analogously to Section~\ref{sec:mainproof}, after each phase we define a feasible solution to DCLP($T$), which gives a lower bound on the total length of the DE-sequences constructed during that phase. 
 The result will be an overall lower bound on the length of the whole DE-sequence, which we then optimize to derive 
Theorem~\ref{thm:2val} below. 
In Subsection~\ref{sec:smallepsilon} we will complete the proof of Theorem~\ref{thm:smallepsilon} by deriving it from Theorem~\ref{thm:2val}.

\begin{thm}\label{thm:2val} Let $\varepsilon < \frac{1}{2}$. Let
  ${\cal I}$ be an instance of
  the $(1,\varepsilon)$-restricted Santa Claus problem
and let $T\in \mathbb{R}$ be such that CLP($T$) has a feasible
solution. Suppose that $1\leq T<  2$, and that
$c :=\lceil T/\varepsilon\rceil\geq 4$. 
Suppose $r\geq 2$ is an integer such that $\sum_{X=r}^c
\frac{1}{a_r(X)} \geq 1$. 
Then there is an allocation with
min-value at least $r\varepsilon$.
\end{thm}

Clearly this theorem is strongest when $r$ is
  largest. We therefore choose $r_c$ to be the largest integer $r \in
\mathbb{N}$ 
satisfying $\sum_{X=r}^c \frac{1}{a_r(X)} \geq
1$. Hence Theorem~\ref{thm:2val} implies the upper bound of $c/r_c$ on the integrality gap for all $c\geq 4$.
In the proof of Theorem~\ref{thm:smallepsilon} we will verify directly that
$c/r_c$ is an upper bound on the integrality gap when $1\leq c\leq 3$ as well.

For convenience, we provide a table showing the triples $(c,r_c,c/r_c)$ for $1\leq c\leq30$ (with $c/r_c$ truncated to two decimal places).

\begin{center}
\smallskip
\begin{small}
\begin{tabular}{|l|l|l|l|l|l|l|l|l|l|l|l|l|l|l|l|l|}
\hline {$c$}&1&2&3&4&5&6&7&8&9&10&11&12&13&14&15\\
\hline {$r_c$}&1&1&1&2&2&2&3&3&4&4&4&5&5&6&6\\
\hline {$c/r_c$}&1&2&3&2 &2.5 &3 &2.33&2.66&2.25&2.5&2.75&2.4&2.6&2.33&2.5\\
\hline
\hline {$c$}&16&17&18&19&20&21&22&23&24&25&26&27&28&29&30\\
\hline {$r_c$}&6&7&7&8&8&8&9&9&10&10&11&11&11&12&12\\
\hline {$c/r_c$}&2.66&2.42&2.57&2.37&2.5&2.62&2.44&2.55&2.4 &2.5 &2.36 &2.45&2.54&2.4&2.5\\
\hline 

 \end{tabular}
\end{small}
\smallskip
\end{center}

In the following subsections we will need to refer to some simple properties of
the pairs $(c,r_c)$. These are spelled out in the following observation.
\begin{obs}\label{obs:crc}
  \begin{itemize}
  \item[(i)] $r_c\geq\frac{c}4$ for every $c\geq 4$,
  \item[(ii)] $c\geq 2r_c+1$ for every $c\geq 5$,
  \item[(iii)] $c\geq 2r_c+2$ for every $c\geq 10$.
  \end{itemize}
\end{obs}
\begin{proof}
The sum 
$A=\sum_{X=r}^c\frac1{a_r(X)}$ has $c-r+1$ terms, that form a
non-decreasing sequence. The smallest term is $\frac1{2r-1}$, implying
that if $c-r+1\geq 2r-1$ then $A\geq 1$. This is easily satisfied by
$r=\lceil\frac{c}4\rceil$, implying (i).

The largest term in $A$ is (at most)
$\frac3{4r-1}$. Note that $\frac3{4r-1}<\frac1{r+1}$ when $r\geq5$, so to
reach 1 in this case there must be at least $c-r+1>r+1$ terms. Hence
$c\geq2r_c+1$ if $r\geq5$. For values of $r\leq 4$ see the table to
complete the proof of (ii).

Similarly $\frac3{4r-1}<\frac1{r+2}$ when $r\geq 8$, so to
reach 1 in this case there must be at least $c-r+1>r+2$ terms. Hence
$c\geq2r+2$ if $r\geq 8$, and again the table completes the proof of (iii) for
the smaller values.
\end{proof}

\subsection{Proof of Theorem~\ref{thm:2val}}\label{ineq}

To prove Theorem~\ref{thm:2val} we will again apply Theorem~\ref{thm:diet} to infer
the existence of the independent transversal in $H(\alpha)$ with an appropriate $\alpha$, which in turn is equivalent to the existence of an allocation of min-value more than $\alpha T$. Hence to infer Theorem~\ref{thm:2val} we need an $\alpha$  satisfying  $r\veps  > \alpha T$, which is one of the conditions on $\alpha$ in Setup~\ref{sec:2val}.

Since we would like to appeal to Lemma~\ref{lem:basedinC}, we will start by verifying that an $\alpha$ satisfying all three conditions in Setup~\ref{sec:2val} exists. To that end it suffices to show that $r\veps < 1$. Indeed, when $c\geq 5$ then 
by Observation~\ref{obs:crc}(ii) we know $c\geq 2r_c+1$,
so by the maximality of $r_c$ we have  $\frac{2}{\varepsilon}+1>\frac{T}{\varepsilon}+1>c\geq 2r+1$. 
Otherwise, when $c=4$ then $r=2$ and the assumption $\varepsilon<1/2$
implies  $r\varepsilon<1$. 

Let us choose, say, $\alpha = (r-0.5)\veps/T$, which satisfies the conditions of Setup~\ref{sec:2val}. 
Then with this $\alpha$, resources of value 1 are fat since $1 > r\veps > \alpha T $. Resources of value $\veps$ are thin because $r\geq 2$ implies $\veps =\alpha T/(r-0.5)<\alpha T$. Thin $\alpha$-hyperedges have size exactly $r$ since $r\veps > \alpha T>(r-1)\veps$.

For the proof we fix a subset $U\subseteq P$, 
assume there is no KO-sequence starting with $J(\alpha)|_U$, and seek
a DE-sequence $\tau$ of length at least $|U| - |F_U|$ starting with
$J(\alpha)|_U$.  Our strategy will be as follows. 
\begin{itemize}
\item[] {\small INITIALIZATION.} Let $\tau$ be a sequence of deletions
  starting with $J(\alpha)|_U$ until no further deletion is possible and
  let $G^*$ be the resulting subgraph. Let $W=\emptyset$. 
  \\[1mm]
  For each $X$, $c\geq X\geq r$, in
 decreasing order, execute the following Phase $X$;
\item[] {\small PHASE $X$:} {\small WHILE} there is a configuration $C$
  with at least $X$    resources remaining in $C\setminus W$, \\
  {\small DO} perform a DE-sequence $\sigma$ starting with $G^*$ based
in $C$ (as given by Lemma~\ref{lem:basedinC} corresponding to the value 
of $X$) and perform all possible deletions afterwards. Update $G^*$ to be the resulting current graph. 
Append $\sigma$ to the end of $\tau$.
Let $W_{\sigma}$ denote the cover of $\sigma$ and set $W: = W \cup W_{\sigma}$.
\end{itemize}

Next we verify that this process is well-defined, that is,
whenever Lemma~\ref{lem:basedinC} is called upon in
some Phase $X$, the conditions of Setup~\ref{sec:2val} are satisfied.

The instance ${\cal I} = (P, R , v, \{L_p : p\in P\})$, target $T\in \mathbb{R}$
and integer $r\geq 2$ are given in the assumptions of
Theorem~\ref{thm:2val}. As indicated above, our choice of $\alpha=(r-0.5)\veps$ ensures that
the conditions on $\alpha$ are satisfied. We have fixed a subset
$U\subseteq P$.

Consider the graph $G^*\subseteq J(\alpha)|_U$ in Phase $X$ to
which we apply Lemma~\ref{lem:basedinC}. Observe that each iteration
of Phase $X$ is immediately preceded by an 
iteration of Phase $X$, of Phase $X+1$, or the initialization
phase. To verify the condition on $W$ in
Setup~\ref{sec:2val}, note that throughout the procedure, the $(\star)$ property
is maintained after each execution of an iteration of Phase $X$ or $X+1$, as
the cover $W_{\sigma}$ of the new segment of $\tau$ is added to
$W$. The $(\star)$ property holds trivially after initialization,
since no explosions 
have yet occurred so $W=\emptyset$ is a cover. Hence $W$ satisfies $(\star)$
with $G_{start}=J(\alpha)|_U$ and $G_{end}=G^*$. Since we have assumed there is
no KO-sequence starting with $J(\alpha)|_U$, there cannot be a
KO-sequence starting with $G^*$. Finally, we check that every edge of
$G^*$ is explodable. Since we end the initialization phase and each
iteration of Phase $X$ or $X+1$ by performing deletions until no more were
possible, we know that all edges of the graph $G^*$ are explodable. Hence the conditions of Setup~\ref{sec:2val} hold.

Let $W_X$ be the union of the covers and $n_X$ be the 
number of explosions done in Phase $X$.
By Lemma~\ref{lem:basedinC},
for each $X$ with 
$c\geq X\geq r$ we have $|W_X|\leq n_Xa_r(X)$.

For $c\geq X \geq r$ we consider the moment 
after the last step in Phase $X$ is 
executed and set $W=W_c\cup\ldots\cup W_X$. 
For each thin configuration $S\in {\cal C}_p(T)$
with $p\in U$, we know that $|S\cap W|\geq c-X+1$,
otherwise we could have continued with another step of Phase $X$. 
Recalling that
  $\varepsilon$ is the common value of all thin resources, we conclude that $v(S\cap W)\geq\varepsilon(c-X+1)$
for each such $S$. Hence we may
  apply Proposition~\ref{prop:dual} with $W$ in place of $Y$ and
  $\varepsilon(c-X+1)$ in place of $c$ to obtain
  $$\varepsilon|W|=v(W)\geq\varepsilon(c-X+1)(|U|-|F_U|).$$

Comparing the upper and lower bounds on $|\cup_{j=X}^c W_j|$ for each
$X = c, c-1 , \ldots , r$, we obtain
$$ \sum_{j=X}^c a_r(j) n_j \geq (c-X+1)(|U|-|F_U|).$$ 

Since the coefficient function $a_r$ is non-increasing in $j$,
in order to minimize the objective function $\sum_{j=r}^c n_j$, we
have to choose $n_c, n_{c-1}, \ldots , n_r$  in reverse order such
that all the inequalities are equalities
$$ \sum_{j=X}^c a_r(j) n_j = (c-X+1)(|U|-|F_U|).$$
This implies that the length $\ell(\tau) = \sum_{j=r}^c n_j$ of the
DE-sequence $\tau$ our process creates is minimized when
$n_j = \frac{|U|-|F_U|}{a_r(j)}$. Consequently
$\ell(\tau) \geq \sum_{j=r}^c \frac{1}{a_r(j)} (|U| -|F_U|)$, which is 
at least $|U| -|F_U|$, as required for Theorem~\ref{thm:diet}. 
This completes the proof of Theorem~\ref{thm:2val}.

\subsection{Proof of Lemma~\ref{lem:basedinC}}\label{sec:shortDE}

The main goal of this subsection is to prove Lemma~\ref{lem:basedinC}. 
To start we describe two criteria that guarantee that a DE-sequence based in $C$ can be continued (plus a consequence of the first one). Both here and in Lemma~\ref{lem:basedinC}, the sets $C$ and $W$ do not play separate roles in the proofs, but appear only in the form $C\setminus W$. To emphasise this we will write $(C\setminus W)$.

\begin{lem}\label{lem:multiexp}
Assume Setup~\ref{sec:2val}. Let $C\in {\cal C}_p(T)$ be a thin configuration with
  $p\in U$ and let $\sigma$ be a 
  DE-sequence starting with $G^*$, 
  with explosions
  of $e_1f_1, \ldots , e_{\ell}f_{\ell}$, in this order, 
  where $e_i\subseteq (C\setminus W)$, $1\leq i \leq \ell$, is owned by $p$. 
\begin{itemize}
    \item[(a)] If $e\subseteq (C\setminus W) \setminus \left(\bigcup_{i=1}^{\ell}
  f_i\right)$ is of size $|e|=r$, then $\sigma$ can be extended to a
longer DE-sequence 
with an explosion involving $e^p$. 
\item[(b)] If $\sigma$ cannot be extended to a longer DE-sequence based in $C$, then 
$av(\sigma)\leq r+\frac{r-1}{\ell(\sigma)}$.
\item[(c)] Let $G$ be the current graph immediately before the $\ell$th
  explosion and suppose $G$ has only explodable edges.
If $e_{\ell}f_{\ell}\in E(G)$ was chosen such that $e_{\ell} \subseteq
(C\setminus W) \setminus \bigcup_{j=1}^{\ell-1}
  f_j$ and $|e_{\ell}\cap
f_{\ell}|$ is maximized with this property, and if $f_{\ell} \cap e_ {\ell} \neq f_{\ell}
  \cap \left((C\setminus W) \setminus \bigcup_{j=1}^{\ell-1}
  f_j\right)$, then $\sigma$ can be extended to a
longer DE-sequence based in $C$.
\end{itemize}
\end{lem}
\begin{proof}
For Part (a), note that being disjoint
from all $f_j$s, the $\alpha$-hyperedge $e^p$ survived all explosions
so far. Since there is no  
KO-sequence starting with $G^*$ isolating $e^p$, after possible
deletions, there will be an explodable edge incident to $e^p$ in the
current graph. So we can extend $\sigma$ with a further explosion
based in $C$ involving $e^p$.

For (b), if $\sigma$ is maximal then
$\left|(C\setminus W) \setminus \bigcup_{j=1}^{\ell}
    f_j \right|\leq r-1$ and for the size of the basic cover we have 
$$|W_{\sigma}| \leq \left|\left(\bigcup_{i=1}^{\ell} f_i\right) \cup
  (C\setminus W)\right| \leq r\ell(\sigma) + r-1.$$  

For Part (c) let us take $y \in (f_\ell \setminus e_{\ell}) \cap \left( (C\setminus W) \setminus \bigcup_{j=1}^{\ell-1} f_j \right) \neq \emptyset$
and $w\in e_{\ell}\setminus f_{\ell} \neq \emptyset$. 
Then the $\alpha$-hyperedge $g = e_{\ell} -w +y$ is contained in $ (C\setminus W) \setminus \bigcup_{j=1}^{\ell-1} f_j $, and
consequently by Part (a) $g^p$ in particular is {\em present} in
the current graph $G$ immediately before the explosion of
$e_{\ell}f_{\ell}$. Since $|f_{\ell} \cap g| > |f_{\ell} \cap
e_{\ell}|$, the $\alpha$-hyperedge $g^p$ was not 
explodable with $f_{\ell}$ immediately before the explosion of
$e_{\ell}f_{\ell}$, when $G$ had only explodable edges.
Hence $g^pf_{\ell}$ was not an edge of $G$ (i.e. it must have been deleted earlier in the sequence $\sigma$ or was already not present in $G^*$).
Therefore $g^p$ survives the $\ell$th explosion as well and since
there is no KO-sequence isolating $g^p$, another explosion involving
$g^p$ is possible.
\end{proof}

We are now ready to prove Lemma~\ref{lem:basedinC}.

\begin{proof} First we assume that $X\leq \frac{3r-1}{2}$.
  If there is a DE-sequence of length two based in $C$,
  then by Lemma~\ref{lem:multiexp}(b) it has average cost
  at most $r+\frac{r-1}{2} \leq 3r - X-1$.

  We may thus assume that there is no DE-sequence of length two based
  in $C$. Since $|C\setminus W| \geq X \geq r$ there exists $e^p \in V(G^*)$ with $e\subseteq (C\setminus W)$ and an explosion involving $e^p$ by Lemma~\ref{lem:multiexp}(a). 
  Among all explodable pairs $e^pf\in E(G^*)$ with $e\subseteq
  (C\setminus W)$, let us choose one with $|e\cap f|$ largest.
If $f\cap e \neq f\cap (C\setminus W)$, then by
Lemma~\ref{lem:multiexp}(c) there is a second explosion based in
$C$, a contradiction. For this recall that $G^*$ has only explodable edges.

Otherwise $f\cap e = f\cap (C\setminus W)$.  
Since no further explosion based in $C$ is possible,
$|(C\setminus W)\setminus f|\leq r-1$ by
Lemma~\ref{lem:multiexp}(a). In other words
$|f\cap (C\setminus W)|\geq X-r+1$, in which case
for the size of the basic cover $f\cup e$ we have $|e| + |f| - |f\cap
(C\setminus W)| \leq 2r -
(X-r+1) =3r-X-1$, as desired.

  We consider now the range $\frac{3r}{2} \leq X$.
We aim to construct a
DE-sequence based in $C$ with average cost at most
$2r-\frac{X+1}{3}$, or one of length
at least $3$, which has average cost at most $\frac{4r-1}{3}$ by
Lemma~\ref{lem:multiexp}(b). This shows that
$av(\sigma) \leq \max \{2r-\frac{X+1}{3}, \frac{4r-1}{3}\}$.
The bound on $av(\sigma)$  in the second and third ranges then
  follows since $2r - \frac{X+1}{3} \geq \frac{4r-1}{3}$ if and
  only if $2r \geq X$.

 For the purposes of this proof we set $b = 2r-\frac{X+1}{3}$.

To begin, first suppose that there exist $g^p\in V(G^*)$  with
$g\subseteq (C\setminus W)$ and
$g^ph\in E(G^*)$ such that  $|h\cup g| \leq  b$. Then since every edge
of $G^*$ is explodable, the single explosion of $g^ph$ is a
DE-sequence of length 1 of the type we seek. We may therefore assume
from now on that $G^*$ contains no such edge.

 We now distinguish two cases.

\medskip

\noindent{\bf Case 1.} 
Suppose first that some $e^p$ with $e\subseteq (C\setminus W)$ has an
explodable neighbor $f$ such that 
  $|(C\setminus W)\cap  f|\leq \lceil\frac{2X-1}{3}\rceil-r = 3r -1 -
  \lfloor 2b \rfloor=: t_1$
  (note that this expression is non-negative for $X$ in our range). \\
We perform this explosion and then deletions until no more are possible. 
Then at least $X-t_1\geq r$ 
resources remain in $(C\setminus W) \setminus f$, so
 by Lemma~\ref{lem:multiexp}(a) there is a further explosion
 possible.  Among all edges $hg^p$ in the current graph, where
 $g\subseteq(C\setminus W) \setminus f$,   we choose one 
 with  $|h\cap g|$ largest.
 Note that $hg^p$ was an edge of $G^*$ as well, and hence we can assume
 $|h\cap g | \leq 2r -1 - \lfloor b \rfloor$.
Indeed,
 otherwise we find
 $|h\cup g| \leq |h| + |g| -|h\cap g| \leq \lfloor b \rfloor
 \leq b$, giving a DE-sequence of length 1 that achieves our aim as
 noted above.

 Now for our second explosion we explode $hg^p$, and follow with a
 sequence of deletions until no more are possible. 

If $h\cap g= h\cap ((C\setminus W) \setminus f)$, then we still have 
$|(C\setminus W) \setminus f|
- |h\cap g| \geq X- t_1 - (2r -1 - \lfloor b \rfloor) \geq
X -5r +2 + \lfloor b \rfloor + \lfloor 2b \rfloor 
\geq X -5r +2 +b+2b -1 
= X-5r +1 +6r - (X+1) = r
$ resources in $(C\setminus W)\setminus (f\cup h)$.
Hence by Lemma~\ref{lem:multiexp}(a) our sequence can be extended to a
sequence of length $3$. 

Otherwise $h\cap g \neq h\cap ((C\setminus W)\setminus f)$, in which case
Lemma~\ref{lem:multiexp}(c) guarantees the extension of our
sequence to a third explosion. 

\medskip

\noindent{\bf Case 2.} Suppose now  that $|(C\setminus W)\cap f|\geq
t_1 + 1$ for 
every explodable neighbor $f$ of an $\alpha$-hyperedge in $C\setminus
W$ owned by $p$. \\
For our first explosion we choose edge $hg^p$ of $G^*$
with $g\subseteq (C\setminus W)$ such that $|g\cap h|$ is largest.
Again, if $|h\cup g|\leq b$ we have achieved our aim with this
DE-sequence of length 1, hence we may assume that its basic cover
$h\cup g$ satisfies $|h\cup g| \geq \lfloor b \rfloor +1$.

If $h\cap g = h\cap (C\setminus W)$
then the set $(C\setminus W) \setminus h$ contains
$X-|h\cap (C\setminus W)| = X - (|h| +|g| -|h\cup g|)
\geq X- (2r - (\lfloor b \rfloor +1)) = 
X - \lceil \frac{X+1}{3} \rceil +1$ resources, which is at least $r$
for $X\geq \frac{3r}{2}$. Consequently we can apply
Lemma~\ref{lem:multiexp}(a) to extend our DE-sequence with a second
explosion involving some $g'\subseteq  (C\setminus W)$ owned by $p$.

If $h\cap g \neq h\cap (C\setminus W)$, then by
Lemma~\ref{lem:multiexp}(c) we can also extend our DE-sequence
with a second explosion involving such a $g'$.

Either way we have a DE-sequence $\sigma$ of length two,
based in $C$, with explosions $gh$ and $g'h'$. 
We claim that the set $W_{\sigma} = h \cup h' \cup 
(C\setminus W) \setminus Y$, where
$Y \subseteq (C\setminus W)\setminus (h \cup h')$ is an {\em arbitrary}
subset of size 
$\min \{t_1, |(C\setminus W)\setminus (h \cup h')| \}$,
is a cover. Indeed, let $f$ be an $\alpha$-hyperedge
disjoint from $W_{\sigma}$, and let us show that $f$ survived the
explosions of $\sigma$.  
Since $f$ is in particular disjoint from $h$ and $h'$, the only way
$f$ could have disappeared is if it had 
an edge of the current graph, and
hence also of $G^*$, to $g$ or $g'$. Recall that $p$ is the owner of
both $g$ and $g'$. But then $f$ is owned by some
$q\neq p$. Since we are in Case 2, the condition $|f\cap
(C\setminus W)| \geq t_1 +1$ holds in particular for $f$. Since $f
\cap W_{\sigma} = 
\emptyset$ we have that $f\cap (C\setminus W) \subseteq Y$, which is a
contradiction since $|Y|\leq t_1$ and hence is too small to contain
$f\cap (C\setminus W)$. 

The DE-sequence $\sigma$ is based in $C$ and has length two, so unless
its average cost is already small enough for our lemma, we have that 
$2b < |W_\sigma| = |h| + |h'|+|((C\setminus W) \setminus
(h\cup h')) \setminus Y|,$ which implies   
$$|(C\setminus W) \setminus (h\cup h')| \geq \lfloor 2b\rfloor +1 -2r + |Y|.$$
We claim that this implies $|(C\setminus W) \setminus (h\cup h')| \geq
r$ and therefore by Lemma~\ref{lem:multiexp}(a) we
can extend $\sigma$ to a third explosion. 
Indeed, if $|Y| = t_1$, then this is immediate from
$t_1=3r-1-\lfloor2b\rfloor$.  
Otherwise we have $ 2r-1 \geq \lfloor 2b \rfloor =
4r-\lceil\frac{2X+2}{3}\rceil$, or equivalently
$\lceil\frac{2X+2}{3}\rceil \geq 2r+1$, which implies $X\geq 3r$.
But then
$|(C\setminus W) \setminus (h\cup h')| \geq |(C\setminus W)|
 - |h| - |h'|= X - 2r \geq r$ as needed.
\end{proof}

\subsection{Proof of Theorem~\ref{thm:smallepsilon}} \label{sec:smallepsilon} 
\begin{proof}[Proof of Theorem~\ref{thm:smallepsilon}]

We define $f$ by
$$f(x)= \frac{1}{xr_{\lceil\frac1x\rceil}}.$$
Here, recall,
that for an integer $c\in \mathbb{N}$, we denote by
$r_c$ the  largest integer $r \in \mathbb{N}$ such that
$\sum_{X=r}^c \frac{1}{a(X)} \geq 1$. 
We show that $f(x)$ bounds the integrality gap for ${\cal I}$, where
$x=\frac{\varepsilon}{T^*}$.

First observe that for any instance with $T^* >0$, it is easy to
check Hall's condition on the bipartite graph of players and coveted
resources to demonstrate that there is a valid allocation of one
resource to each player. Hence, in the two-values case
$OPT \geq \varepsilon$.
Since $T^* \geq OPT$, we always have that the integrality gap is at
most $\frac{T^*}{\varepsilon} \geq 1$. This shows that $f(x) = \frac{1}{x}$
is an appropriate choice to bound the integrality gap for every $x\in (0,1]$.
Since $r_1=r_2=r_3 =1$, this verifies the statement when
$x\geq\frac{1}{3}$.

We proceed now with the case $x := \frac{\varepsilon}{T^*} < \frac13$. 

Analogously to \cite{chantangwu}
we start by reducing the problem to the case when
  $1\leq T^* < 2$. 
If $T^*\geq 2$ recall that
the additive approximation result of Bez\'akov\'a and
Dani~\cite{bezakovadani}, mentioned in the Introduction,
gives a polynomial-time algorithm to find an allocation with min-value
$T_{ALP}- \max v_r$, where $T_{ALP}$ is the optimum of the standard
assignment LP. Hence $OPT\geq T_{ALP}-1 \geq T^*-1 \geq
\frac{T^*}{2}$ as the CLP is stronger than
the ALP and in our case $\max v_r =1$. So the integrality gap is
at most $2$, which is less than $f(x)$ for every $x$.

If $T^* < 1$, then we create another instance ${\cal I}'$  where for each
resource $r\in R$ with $v_r=\varepsilon$, we change the value of $r$ to
$v_r'=\varepsilon' : = \frac{\varepsilon}{T^*}$,
and keep $v_r'=v_r=1$ otherwise. It is easy check that $T^*({\cal I}')
=1$ and $OPT({\cal I}) \geq OPT({\cal I}')T^*$. Note that
$\varepsilon' < \frac{1}{3}$.
Hence, applying our theorem to
the $(1, \varepsilon')$-instance ${\cal I}'$, there is
an allocation with min-value at least
$\frac{T^*({\cal I}')}{f(\varepsilon'/T^*({\cal I}'))}$.
The very same  allocation in ${\cal I}$ has min-value at least
$\frac{T^*}{f(\varepsilon')}$ and hence exhibits an
integrality gap of at most
$f(\varepsilon/T^*)$ for  ${\cal I}$.

From now on we assume $1\leq T^* < 2$ and
apply Theorem~\ref{thm:2val}. For this we note that
$c:= \lceil\frac{T^*}{\varepsilon}\rceil = \lceil \frac1x\rceil > 3$
and $\varepsilon < \frac{1}{2}$.
Theorem~\ref{thm:2val} then implies that we have an
allocation for ${\cal I}$ with min-value at least $r_c\varepsilon$, thus
$OPT\geq r_c\varepsilon$. Hence the integrality gap for ${\cal I}$ is
at most

$$\frac{T^*}{r_c\varepsilon}=\frac{1}{xr_{\lceil\frac1x\rceil}}=f(x).$$
as promised.

To prove the assertions of Theorem~\ref{thm:smallepsilon} about the
values, first note that when $x\geq \frac13$ we have $f(x) = \frac{1}{x}$,
which is less than $3$ unless $x=\frac13$, and at most $2.75$ for
$x\geq \frac{4}{11}$.
For $x < \frac{1}{3}$ note that with $c=\lceil \frac{1}{x}\rceil$,
we have $f(x)=\frac1{xr_c}\leq\frac{c}{r_c}$. It is easy to verify 
that  for every $c\geq 4$ we have $c/r_c \leq 2.75$, unless $c=6$.
(In fact the ratio $2.75$ is attained on the unique pair $c=11, r=4$.) 
That is, unless $\lceil \frac{1}{x}\rceil =6$ we have $f(x) \leq 2.75$. 
When $\lceil \frac{1}{x}\rceil =6$ we have $x\geq \frac{1}{6}$ and
$f(x)=\frac{1}{r_6x} = \frac1{2x}$, which 
is at most $2.75$ for $x \geq \frac{2}{11}$
and strictly less than $3$ unless $x=\frac16$.

Finally, we deal with the case in which  $x \rightarrow 0$.
Then $c=\lceil\frac1x\rceil\rightarrow\infty$, and hence also
$r_{c} \geq \frac{c}{4}\rightarrow\infty$ by
Observation~\ref{obs:crc}(i).  

Let us assume
that $x < \frac{1}{10}$, so $c \geq
10$ implying that $c \geq 2r_{c}
+2$ by Observation~\ref{obs:crc}(iii). Setting $r:=r_c$, we write the sum as
$ \sum_{X=r}^c \frac{1}{a(X)} = A_r + B_r +C_r,$ where
\begin{align*} A_r:=&\sum_{X=r}^{\lfloor (3r-1)/2\rfloor}\frac{1}{3r-X-1} =
\sum_{k=\lceil (3r-1)/2\rceil}^{2r-1}\frac1k = (H_{2r-1}-H_{\lceil
                      (3r-1)/2\rceil -1})\rightarrow \ln(4/3),\\
                      B_r:= & \sum_{\lceil 3r/2\rceil}^{2r}\frac{3}{6r-X-1} =
\sum_{k=4r-1}^{\lfloor \frac{9r}{2}-1\rfloor}\frac3k= 3(H_{\lfloor
                              \frac{9r}{2}-1\rfloor}-H_{4r-2})\rightarrow
                              3\ln(9/8),\\
                              C_r : = & \sum_{X=2r+1}^c \frac{3}{4r-1}
                                        = \frac{3(c-2r)}{4r-1},
\end{align*}
when $r \rightarrow \infty$.
Here we use the well-known fact for the harmonic 
series $H_n=\sum_{k=1}^n\frac1k$, that $H_n-\ln n$
converges to a constant. 

By the maximality of $r_{c}$ and using $c
\geq 2r_{c} +2$, we have that 
$$A_{r_{c}+1} + B_{r_{c}+1} + \frac{3(c
  -2(r_{c}+1))}{4(r_{c}+1)-1} < 1.$$
Recall that when $x \rightarrow 0$, we also have $c
\rightarrow \infty$ and $r_{c} \rightarrow
\infty$, so we obtain
$\ln (4/3) + 3 \ln (9/8) + \frac{3}{4} \lim_{x \rightarrow 0}f(x) - \frac{3}{2}
\leq 1$, where we again use that $f(x) \leq \frac{c}{r_c}$.
Hence $$\lim_{x \rightarrow 0} f(x) \leq
\frac{10}{3} - \frac{4}{3} \ln(4/3) - 4\ln
(9/8) < 2.479,$$
as desired.
\end{proof}

\section{Conclusion}

In this paper we give an entirely novel approach, based on topological
notions, for bounding the integrality gap of the Santa Claus problem. 
This leads to significant improvements on the best known estimates. 
We believe that this approach will prove to be fruitful in addressing
other algorithmic problems involving hypergraph matchings.

As mentioned in the introduction, our argument at the moment does
not come with an efficient algorithm for finding an allocation with
the promised min-value. This is primarily due to the fact that
we do not have a good upper bound on the number of simplices in
the triangulation described in the Appendix, which ultimately
governs the running time of any algorithmic procedure
based on our argument. It would be of great interest to develop
methods to make the approach more efficient.

A possible ray of hope comes from recalling the eventual success of turning the
initially highly ineffective combinatorial procedure of
\cite{asadpourfeigesaberi}, based on \cite{haxell}, 
into an efficient algorithm with the same constant factor
appriximation. This was achieved through a series
of important contributions of several authors, as described in the introduction. 
Even a quasipolynomial-time algorithm based on our approach that provides
{\em any} constant factor approximation would seem to require new
ideas.
Such an algorithm would be a first step towards an efficient
approximation algorithm that breaks the factor $4$ barrier.

Finally, we would like to recall from the introduction
that our work on the $(1,\varepsilon)$-restricted case identifies certain
parameter choices that seem to capture a key difficulty for the
CLP-approach.
Specifically, we would like to see a $(1,1/3)$-restricted problem
instance that has optimal CLP-target $T^*=1$, and no allocation of
min-value $2/3$.

\medskip
\noindent{\bf Acknowledgements:} The authors would like to thank Lothar
Narins for helpful discussions in the early stages of this work, and Olaf Parczyk and Silas Rathke 
for the data on $r_c$  (Section~\ref{sec:2val}) and the figure for $a_r$ (Section~\ref{ar}), respectively.
They are also very grateful to the anonymous referees for many insightful comments and helpful suggestions.

\section{Appendix}
\subsection{The parameter $\eta$}
As mentioned in Section~\ref{sec:topological} for our results we need only that there
exists a graph parameter $\eta$ that satisfies Fact~\ref{fact:eta} and
Theorems~\ref{thm:aharoniberger} and \ref{thm:meshulam}. In fact such
a parameter can be defined in a purely combinatorial way, without any
explicit reference to topology. For completeness we begin with a
precise definition of $\eta$ following the
treatment of \cite{haxell2}, where the required properties are
verified. 
However, the intuition behind the parameter $\eta$ and how we use it
in our proofs is very much topological, as we describe after the definition. 

\paragraph{The definition.}
An {\em abstract
  simplicial complex} is a set $\cA$ of subsets $A$ of a finite set
$V=V(\cA) = \cup_{A\in {\cal A}} A$ with the property that if $A\in\cA$ 
and $B\subset A$ then $B\in\cA$. We call the sets $A$ the {\it
  simplices} of $\cA$, and the {\it dimension} of $A$ is $|A|-1$.
The {\it dimension} of $\cA$ is the maximum dimension of any
$A\in\cA$. 
Let ${\cal A}$ and $\Sigma$ be abstract simplicial complexes. 
A function $f:V(\cA) \rightarrow V(\Sigma)$ is called a {\em simplicial map
from ${\cal A}$ to $\Sigma$} if $f(A) \in \Sigma$ for every $A\in \cA$.
We say that $\cA$ is a {\em $d$-PSC}, i.e. a {\it pure}
$d$-dimensional simplicial complex, if every maximal $A\in\cA$ has the same
dimension $d$. 
Note then that a $d$-PSC is the {\it closure} of the ($d+1$)-uniform
hypergraph $\cA^d$ consisting of the $d$-dimensional simplices of
$\cA$, that is, 
we form $\cA$ from $\cA^d$ by adding all subsets of the hyperedges of
$\cA^d$.
For a $d$-PSC $\cA$, the {\it boundary}  $\partial\cA$  of $\cA$ is
the ($d-1$)-PSC that is the closure of the $d$-uniform hypergraph
$$\{B:|B|=d, |\{A\in\cA^d:B\subset A\}|\equiv 1\mod 2\}.$$
If $\partial(\cA)$ is empty we say that $\cA$ is a {\em $d$-dimensional $Z_2$-cycle}.
The abstract simplicial complex $\Sigma$ is said to be $k$-{\it connected} if
for each $d$, $-1\leq d\leq k$, for every $d$-dimensional $Z_2$-cycle
$\cA$ and every
simplicial map $f:\cA\to\Sigma$, there exists a $(d+1)$-PSC $\cB$ and
a simplicial map  $f':\cB\to\Sigma$ such that $\partial\cB=\cA$
and the restriction $f'|_{\cA}$ of $f'$ to $\cA$ satisfies 
$f'|_{\cA}=f$. 

The independence complex of a graph is an abstract simplicial complex
and the value of $\eta (G)$ for a graph $G$ is defined as the
largest integer $t$ such that the independence complex
${\cal J}(G)$ is $(t-2)$-connected. The parameter $\eta$ is not
explicitly defined in \cite{haxell2}, but the main theorems about
$\eta$ are stated and proved there in terms of the above
definition of $k$-connected. (Theorem~\ref{thm:aharoniberger} and
\ref{thm:meshulam} appear 
as Theorems 11 and 12, respectively.) We may verify Fact~1 as
follows. Fact~1(1) follows directly from 
the definition of $\eta$, as saying that $\Sigma$ is $(-1)$-connected is the 
same as saying that $V(\Sigma)$ is nonempty. For the second statement of
Fact~1(2), suppose $G$ contains an isolated vertex $x$. Let ${\cal A}$ be an
arbitrary $Z_2$-cycle, with a simplicial map $f$ from ${\cal A}$ to ${\cal
  J}(G)$. Then $f$ can be extended to a simplicial map from the
closure ${\cal B}$ of
$\{ A\cup \{w \} : A \in {\cal A}\}$, where $w\not \in
V({\cal A})$ is a new vertex, by mapping $w$ to $x$.  Since the
dimension of ${\cal A}$ is arbitrary, this implies that $\eta (G)$ is
infinite. Otherwise $G$ has no isolated vertices, and so $G_2$ contains an edge.  Then by Theorem~\ref{thm:meshulam} we can keep
deleting  and/or exploding edges from $G_2$, one by one, until all
edges of $G_2$ have disappeared. The resulting graph $G_{end}$ still
contains $G_1$. If $G_{end}$ has an
isolated vertex, then $\eta (G) \geq \eta (G_{end}) = \infty$ by the
above. Otherwise at least one explosion was performed and
$G_{end}=G_1$, hence $\eta (G) \geq \eta (G_1) + 1$ 
by Observation~\ref{obs:sequence}(i).

\paragraph{The intuition.}
In what follows, we describe the topological nature
of our work at an intuitive level, without getting into precise
details.
The topological space $X$ is said to be {\em $k$-connected} if for every $d$,
$-1 \leq d \leq k$, every continuous map from the $d$-dimensional
sphere to $X$ extends to a continuous map from the $(d+1)$-dimensional
ball to $X$. This property indicates that $X$ lacks a $(d+1)$-dimensional
''hole''.

To get a better understanding for the topological core of our
arguments, it helps to think of connectedness as defined in the
following way, that provides a link between the notion of
connectedness for a topological space and our earlier definition for
abstract simplicial complexes. This link goes via triangulations of a
simplex, which are {\em geometric simplicial complexes}, that can be
viewed both as topological spaces and as abstract simplicial complexes.
We say that an abstract simplicial complex $\Sigma$ is
{\em $k$-connected} if for every $d$, $-1 \leq d\leq k$, for every
triangulation ${\cal T}$ of the boundary of the $(d+1)$-dimensional
simplex $\tau$, and every simplicial map $f$ from ${\cal T}$ to
$\Sigma$, there exists a triangulation ${\cal T}'$ of the whole of
$\tau$ that extends 
${\cal T}$, and a simplicial map $f'$ from ${\cal T}'$ to $\Sigma$ that
extends $f$.

Our argument gives a process that, given an instance ${\cal I}$ of the
Santa Claus problem with player set $P$, produces an allocation with
the promised min-value. Very broadly speaking, the process has two
main stages. Following the proofs of Theorems~\ref{thm:aharoniberger}
and \ref{thm:meshulam}, the first stage constructs a triangulation
${\cal T}$ of the 
$(|P|-1)$-dimensional simplex $\tau$, and a simplicial map $f$ from
${\cal T}$ to the independence complex of the graph $H(\alpha)$
(defined in Section~\ref{sec:setup}), such
that the $|P|$-coloring of the points $v \in V({\cal T})$, defined by
the ''owner'' of 
the $\alpha$-hyperedge $f(v)$, satisfies the conditions of Sperner's
Lemma. The second stage applies Sperner's Lemma to find a multicolored
simplex, which corresponds to an independent transversal of
$H(\alpha)$, i.e. an allocation for instance ${\cal I}$ with
min-value at least $\alpha T$ as promised.

Executing the first stage is the main aim of our paper and here is where
topological connectedness helps us. 
The triangulation ${\cal T}$ and the map $f$ are built on the faces of
$\tau$ one by one, in increasing order of dimension.
When ${\cal T}$ and $f$ on a face
$\sigma$ of dimension $d$ are to be defined, triangulations and
maps of all the facets of $\sigma$ are already in place, forming the
boundary of $\sigma$, and these 
need to be extended to a triangulation and a map of the whole of
$\sigma$. 
This notion of extending a map from the boundary of $\sigma$ to the interior is
captured by the parameter $\eta$, so if $\eta$ is sufficiently 
large for each $\sigma$, then this extension is possible. 

\subsection{Demonstrating DE-sequences} \label{sec:example}

Here we demonstrate how to use DE-sequences to show that $\eta(G)\geq
2$ for the cycle $G=C_5$ of length 5. (In fact $\eta(G)=2$, since the
independence complex itself is a $5$-cycle, which has a
$2$-dimensional hole.) We do this without thinking about the underlying ''topology'' of the current graph (i.e. whether the next edge is deletable or explodable), but rather following through {\em all} sequences of deletions and explosions that are possible combinatorially (some of which might not actually represent a ''topologically legal" DE-sequence) and arriving at a lower bound of at least $2$ for each. 

For an edge
$e$ of $G$, the graph $G\explode e$ consists
of a single isolated vertex, and hence $\eta(G\explode e)=\infty$ by
Observation~\ref{obs:sequence}(ii). Therefore if $e$ is  
explodable we are done, so we may assume that $e$ is deletable. Deleting $e$ results in the path
$P_5$ with 5 vertices, and by the definition of deletable edge
$\eta (G) \geq \eta (P_5)$. Next
consider an edge $e'$ of $P_5$ joining two of its degree-2
vertices. Again $P_5\explode e'$ consists
of a single isolated vertex, showing that we may assume $e'$ is not explodable and
hence deletable. The graph $P_5-e'$ consists of two components, a
$P_2$ and a $P_3$. Each of these has a positive value of $\eta$ by
Fact~\ref{fact:eta}(1). Hence $$\eta(G)\geq \eta (P_5)
\geq \eta(P_5-e')\geq 1 + \eta (P_3) \geq 2$$ by
Fact~\ref{fact:eta}(2). 

\subsection{The function $a_r$}\label{ar} 

\begin{figure}[h!tbp]
	\centering
	\begin{tikzpicture}[scale=.35]
		
		\newcommand{\drawfunctionincolour}[3]{
			\pgfmathtruncatemacro{\firstend}{(3*#1-1)/2}
			\pgfmathtruncatemacro{\secondstart}{ceil(3*#1/2)}
			\pgfmathtruncatemacro{\secondend}{2*#1}
			\pgfmathtruncatemacro{\thirdstart}{2*#1+1}
			\draw [#2, thick, domain=#1:\firstend] plot ({\x}, {3*#1-\x-1});
			\draw [#2, thick, domain=\secondstart:\secondend] plot ({\x}, {2*#1-(\x+1)/3});
			\draw [#2, thick, domain=\thirdstart:#3] plot ({\x}, {(4*#1-1)/3});
		}
	
		\draw[step=1cm,gray,very thin, opacity=.4] (0,0) grid (15,15);
		
		\draw[thick,->] (-.2,0) -- (15,0) node[anchor=north west] {$X$};
		\draw[thick,->] (0,-.2) -- (0,15) node[anchor=south east] {$y$};
		
		\drawfunctionincolour{3}{blue}{15}
		\drawfunctionincolour{4}{red}{15}
		\drawfunctionincolour{5}{orange}{15}
		\drawfunctionincolour{6}{purple}{15}
	\end{tikzpicture}
	\caption{Plots of the function $a_r(X)$ for different values of $r$. The function for $r=3$, $r=4$, $r=5$ and $r=6$ is given in blue, red, orange and purple respectively}
\end{figure}

\FloatBarrier

\subsection{Notation Finder}

For each term we indicate the section in which is defined. Most of the definitions appear very close to the beginning of the relevant section or subsection.

\begin{itemize} 
\item $(1,\veps)$-restricted problem~(\ref{sec:introduction})
\item $\alpha$-hyperedge of $p$~(\ref{sec:setup})
\item $\alpha$-edge~(\ref{sec:proofs})
\item allocation~(\ref{sec:introduction})
\item $a_r(X)$~(\ref{sec:2val})
\item $av(\sigma)$, the average cost of $\sigma$~(\ref{sec:2val})
\item block~(\ref{sec:warmup})
\item CLP($T$), the configuration LP with target $T$~(\ref{sec:introduction})
\item $C_p(T)$, the set of configurations for $p$~(\ref{sec:introduction})
\item cheap DE-sequence~(\ref{sec:economical})
\item cover~(\ref{sec:strategy})
\item DCLP($T$), the dual of the CLP($T$)~(\ref{sec:warmup})
\item deletable~(\ref{sec:topological})
\item DE-sequence~(\ref{sec:topological})
\item $e-x$, $e+x$~(\ref{sec:proofs})
\item explodable~(\ref{sec:topological})
\item explodable at resource $r$~(\ref{sec:proofs})
\item $e^p$, hyperedge owned by $p$~(\ref{sec:setup}),~(\ref{sec:proofs}) 
\item $\eta(G)$~(\ref{sec:topological})
\item $F$, the set of fat resources, $F_U$~(\ref{sec:strategy})
\item fat~(\ref{sec:strategy})
\item $f(x)$~(\ref{sec:smallepsilon})
\item $G-e$ ($G$ delete $e$), $G\explode e$ ($G$ explode $e$)~(\ref{sec:topological})
\item $H(\alpha)$, the $\alpha$-approximation allocation graph~(\ref{sec:setup})
\item $H_n$, the harmonic series~(\ref{sec:smallepsilon})
\item $i/j$-DE-sequence~(\ref{sec:economical})
\item independent transversal~(\ref{sec:setup})
\item $\cI$, an instance~(\ref{sec:introduction})
\item integrality gap~(\ref{sec:introduction})
\item $J(\alpha)$, $J(\alpha)|_U$~(\ref{sec:strategy})
\item $\cJ(G)$, the independence complex of $G$~(\ref{sec:topological})
\item KO-sequence~(\ref{sec:topological})
\item $\ell(\sigma)$, the length of $\sigma$~(\ref{sec:topological})
\item $L_p$, the liked set of $p$~(\ref{sec:introduction})
\item $m=\alpha T$~(\ref{sec:warmup})
\item maximal cheap DE-sequence~(\ref{sec:mainproof})
\item min-value~(\ref{sec:introduction})
\item $n_j$~(\ref{sec:mainproof})
\item $n_X$~(\ref{ineq})
\item $OPT$~(\ref{sec:introduction})
\item owner~(\ref{sec:setup})
\item $P$, the set of players~(\ref{sec:introduction})
\item $R$, the set of resources~(\ref{sec:introduction})
\item $r_c$~(\ref{sec:2val})
\item survives~(\ref{sec:proofs})
\item $T$ (target), $T^*$ (optimal target), $T_{ALP}$ (assignment LP optimum)~(\ref{sec:introduction})
\item thin~(\ref{sec:strategy})
\item $v$, $v_r$, $v(S)$, the value function~(\ref{sec:introduction})
\item $W$, a cover, satisfying ($\star$)~(\ref{sec:strategy})
\item $W_j$~(\ref{sec:mainproof})
\item $W_X$~(\ref{ineq})
\item $x=\lceil\frac{\veps}{T^*}\rceil$~(\ref{sec:smallepsilon})
\item $Y_{\leq d}$, $Y_{>d}$~(\ref{sec:mainproof})
\end{itemize}

\end{document}